\documentclass[12pt, draftclsnofoot, onecolumn]{IEEEtran}
\usepackage{caption}
\usepackage{amssymb, amsmath, amsthm}
\usepackage{amsfonts, bbm}
\usepackage{graphicx}
\usepackage{epstopdf}
\usepackage{times}
\usepackage{float} 
\usepackage{array}
\usepackage{arydshln}
\usepackage{bibentry}
\usepackage{booktabs}

\usepackage{cite}
\usepackage{subfigure}
\usepackage{enumitem}
\usepackage{color}
\usepackage{setspace}
\usepackage{bigstrut}
\usepackage{multirow}

\allowdisplaybreaks
\theoremstyle{plain}

\theoremstyle{plain}
\newtheorem{corollary}{Corollary}
\theoremstyle{plain}
\newtheorem{lemmacounter}{Theorem}
\newtheorem{lemma}[lemmacounter]{Lemma}
\theoremstyle{plain}

\theoremstyle{plain}

\newtheorem{special case}[defcounter]{Special Case}
 
\newcommand*{\Scale}[2][4]{\scalebox{#1}{$#2$}}%

\newcommand{\figref}[1]{Fig.~\ref{#1}}
\begin{document}

\title{Coverage and Rate Analysis for Co-Existing RF/VLC Downlink Cellular Networks}
\author{Hina Tabassum and Ekram Hossain\thanks{The authors are with the Department of Electrical and Computer Engineering at the University of Manitoba, Canada.
 (emails: \{Hina.Tabassum, Ekram.Hossain\}@umanitoba.ca). This work was supported by the Natural Sciences and Engineering Research Council of Canada (NSERC). 
}}
\maketitle

\begin{abstract}
Integrating visible light communication (VLC) with conventional radio frequency (RF)-enabled wireless networks has been shown to  improve the achievable data rates of mobile users. 
This paper provides a stochastic geometry framework to perform the coverage and rate analysis of a typical user in co-existing VLC and RF networks covering a large
indoor area. The developed framework can be customized to capture the performance of a typical user in various network configurations such as (i)~{\em RF-only}, in which only small base-stations (SBSs) are available to provide the coverage to a user, (ii)~{\em VLC-only}, in which only optical BSs (OBSs) are available to provide the coverage to a user, (iii)~{\em opportunistic RF/VLC}, where a user selects the network with maximum received signal power, and (iv) {\em hybrid RF/VLC}, where a user can simultaneously utilize the available resources from both RF and VLC networks. The developed model for VLC network precisely captures the impact of the field-of-view (FOV) of the photo-detector (PD) receiver on the number of interferers, distribution of the aggregate interference, association probability, and the coverage of a typical user. Closed-form approximations are presented for special cases of practical interest and for asymptotic scenarios such as when the intensity of SBSs becomes very low. The derived expressions enable us to obtain closed-form solutions for various network design parameters (such as intensity of OBSs and SBSs, transmit power, and/or FOV) such that the number of active users can be distributed optimally among RF and VLC networks. Also, we optimize the network parameters in order to prioritize the association of users to VLC network. Finally, simulations are carried out to verify the derived analytical solutions. It is shown that the performance of VLC network depends significantly on the receiver's FOV/intensity of SBSs/OBSs and careful selection of such parameters is crucial to harness the benefits of VLC networks. Important trade-offs between height and intensity of OBSs  are highlighted to optimize the performance of a user in VLC networks.
\end{abstract}

\newpage

\begin{IEEEkeywords}
Multi-cell downlink visible light communication (VLC) networks,  rate coverage probability, field-of-view, traffic load distribution, Poisson Point Process (PPP).
\end{IEEEkeywords}

\section{Introduction}

The scarcity of spectrum, cost, and interference  in traditional radio frequency (RF) spectrum pushes the network operators to exploit higher frequencies such as visible light communication (VLC) for cellular transmissions~\cite{1,2,3}.  VLC offers (a)~significantly higher transmission capacity due to wider modulation bandwidths, (b)~relatively secure transmissions and less susceptibility to electromagnetic interference due to higher penetration losses, (c)~exhaustive reuse of frequency,  (d)  reduced cost of wireless communication due to the unregulated  spectrum, and (e) power saving since the VLC transmitters can be used both for illumination and communication. VLC possesses a number of interesting features such as higher data rate and spectral efficiency, higher energy efficiency, lower battery consumption and latency  to address the requirements of  evolving 5G/B5G systems~\cite{feng2016applying}. Recently, four VLC standards have been developed that include Japan Electronics and Information Technology Industries Association (JEITA) CP-1221, JEITA CP-1222, JEITA CP-1223 and IEEE 802.15.7~\cite{khan2016visible}.

In VLC, white light is generated  using the wavelength converters and light emitting diodes (LEDs) on the transmitter side. Typically, the generation of white light using LED is trichromatic (red, green and blue) which ensures higher bandwidth as well as higher data rates. 
The modulation in VLC differs from that of RF due to the non-encoding feature of information in phase and amplitude of the light signal. Modulation in VLC is achieved using variations in the intensity of the light corresponding to the information in the message signal. 
In the VLC receiver, the light is detected  and then converted to photo current. VLC is vulnerable to interference from other light sources such as sunlight and other LEDs; therefore, optical filters are needed to mitigate the DC noise components present in the received signal. Modern VLC systems based on intensity modulation (IM) and direct detection (DD) with
optical orthogonal frequency division multiplexing (OFDM)
have been shown to achieve data rates in the range of Gbps~\cite{gbps1,gbps2}.

Although both VLC and RF transmissions use electromagnetic radiation for the information transfer, their properties differ significantly. The wavelength of the visible spectrum (380 nm to 750 nm) is much smaller than the  area of a photodetector (PD) receiver, which effectively removes {\em multi-path fading} (an important channel attenuation factor in RF transmissions). Further, optical signals do not interfere with the RF electronic systems and can be used in sensitive areas such as hospitals~\cite{hospital1} and aircrafts henceforth. Contrary to RF, VLC is susceptible to indoor blockages (walls, human, material objects, etc.) thus naturally confined to a small area.
The received power at PD  depends heavily on the  line of sight (LoS) signals that may get blocked due to limited field-of-view (FOV) of the PD receivers, and/or  radiance angle of the optical LEDs. As such, the dense deployment of  LEDs may not guarantee a reliable coverage. VLC is thus considered as a complimentary rather than substituting technology to RF~\cite{hybconf1,hybconf2, bao2014protocol}. 


\subsection{Background Work}

Recently, few research proposals have investigated the performance of hybrid RF/VLC systems~\cite{kashef2016energy,li2016mobility,wang2017load,li2015cooperative}. In \cite{kashef2016energy}, energy efficiency of an indoor network composed of a single RF base-station (BS)) and a single VLC BS has been maximized constrained by the required data rates for the users and the maximum allowable transmission powers for the BSs. Users are capable of receiving data from both VLC and RF communication systems. The energy efficiency of hybrid RF/VLC communication system is compared to that of the RF-only system. 
\cite{li2016mobility} proposes a mobility-aware load balancing scheme, which dynamically associates users to their corresponding BSs. Applying matching theory, the association problem is formulated and solved  as a college admission problem (CAP). In \cite{wang2017load}, evolutionary game theory-based load balancing algorithm is proposed for hybrid VLC/RF networks considering channel blockage and shadowing. The proposed scheme is shown to improve user satisfaction levels at reduced computational complexity. In \cite{li2015cooperative}, cooperative load balancing scheme with proportional fairness is proposed. 
Both centralized and distributed resource-allocation algorithms are developed. Results demonstrate that the proposed scheme provides a higher area spectral efficiency (ASE) with reasonable fairness. \cite{bao2017visible} proposed a
hybrid VLC/RF heterogeneous network (VLC-HetNet) where VLC and RF channels are  used for downlink and uplink transmissions, respectively. New VLC
frame, multi-user access mechanism, horizontal and vertical handover protocols are discussed.

In \cite{rakia2016optimal}, outage analysis of a dual-hop VLC/RF data transmission system is considered with energy harvesting. The energy carried by the DC component of the received optical signal is harvested for data retransmission at the relay instead of discarding it. The DC bias is optimized at the LED to maximize the overall transmission rate. 
Another relevant study is \cite{chen34downlink} where four different cellular network models (e.g., square, hexagonal, Poisson Point Process (PPP), Matern Hard core process [MHCP] models) for VLC networks are considered and the signal-to-interference-plus-noise ratio (SINR) outage of a typical user is derived.  It is shown that PPP-based cellular model is the most appropriate and tractable for the performance analysis in an indoor environment with multiple {\em attocells}\footnote{An {\em attocell} refers to the coverage area of an optical base station (OBS).}. 
Another interesting and very recent work on hybrid mm-wave and VLC network is conducted by \cite{shao2015design} where multiple VLC BSs and RF BSs are considered. In particular, this study quantifies the minimum spectrum and power requirements for RF network  to achieve certain per user rate coverage performances.

\subsection{Paper Contributions}
Complementary to the aforementioned works, the contributions of this paper can be summarized as follows:
\begin{itemize}
\item We consider a co-existing VLC and RF network covering a large
indoor area. The developed framework is unified to capture the performance of a typical user  in various network configurations such as (i)~{\em RF-only} in which only small BSs (SBSs) are available to provide the coverage to a user, (ii)~{\em VLC-only} in which only optical BSs (OBSs) are available to provide the coverage to a user, (iii)~{\em opportunistic RF/VLC} where a user selects the network with maximum received signal power, and (iv) {\em hybrid RF/VLC} where a user can  utilize the available resources from both RF/VLC networks.

\item The developed framework  is  precise in terms of capturing the impact of field-of-view (FOV) of the photo-detector (PD) receiver on the number of interferers, distribution of the aggregate interference, association probability, and the coverage of a typical user. Typically, the received signal as well as the interference model considers full FOV of $180^\circ$ at the PD receiver which leads to a worst-case bound on the interference~\cite{chen34downlink,shao2015design}. Nonetheless, for a typical user, the received interference power as well as the received signal power depends heavily on the FOV of the receiver.  Moreover, considering a full FOV of the PD receiver guarantees that the serving BS is always within the FOV of the desired receiver, which may not be true and hence  overestimate the performance of the optical network.

\item For a multi-cell VLC network, we first derive the necessary condition for a typical user to have an OBS within its FOV. We then derive the conditional Laplace Transform of the aggregate interference from OBSs in closed-form. Then we derive the exact coverage probability and rate using Gil-Pelaez inversion theorem~\cite{gil} and Hamdi's lemma~\cite{hamdi2010useful}, respectively. The approach is different from \cite{chen34downlink,shao2015design}  where the intensity of interference is first approximated using the moments of the interference and the coverage probability is then approximated in terms of one infinite integral, one infinite summation as well as two cascaded summations. 
{\em This work provides exact coverage and rate with double integrals}.
We also present approximate coverage probability for certain special cases.

\item For opportunistic RF/VLC networks, we first derive the  selection probability of an optical BS as well as RF SBS by a typical user. We then derive the distance distribution for the selected optical/RF BS and then determine the coverage probability of a typical user.  Closed-form expressions and approximations are presented for special cases of practical interest such as when the intensity of SBSs is asymptotically low. The derived expressions enable us to obtain closed-form solutions for various network parameters (such as intensity of OBSs, SBSs, transmit power, and/or FOV) in order to distribute the traffic load among different networks optimally. Also, we optimize the network parameters in order to prioritize the association of users to VLC network. 

\end{itemize}
Simulations are carried out to verify the derived analytical solutions. It is shown that the performance of VLC network depends significantly on the receiver's FOV/intensity of SBSs/OBSs and careful selection of such parameters is crucial to harness the benefits of VLC networks. 


\subsection{ Paper Organization and Notations}
\subsubsection{Paper Organization} The rest of the paper is organized as follows. System model is presented in Section~II along with the
description of the RF/VLC channel models. The exact coverage probability and rate  analysis of a typical user in isolated RF and VLC network is detailed in Section~III.  The exact coverage probability and rate are then derived for a typical user considering opportunistic RF/VLC and hybrid RF/VLC networks in Section~IV.
Numerical results are presented in Section~V and possible extensions to the framework are discussed in Section~VI followed by the conclusion in Section VII.

\subsubsection{Notations}
$\mathrm{Gamma}(\kappa_{(\cdot)},\Theta_{(\cdot)})$ denotes Gamma distribution with shape parameter $\kappa$, scale parameter $\Theta$ and $(\cdot)$ is the name of the random variable (RV). $\Gamma(a)=\int_0^\infty x^{a-1} e^{-x} dx$ is the Euler Gamma function, ${\Gamma}_u (a;b)=\int_b^\infty x^{a-1} e^{-x} dx$ is the upper incomplete Gamma function, and ${\Gamma}_l (a;b)=\int_0^b x^{a-1} e^{-x} dx$ is the lower incomplete Gamma function. $f_X(\cdot)$, $\mathbb{P}_X(\cdot)$, and $\mathcal{L}_X(\cdot)$ denote the probability density function (PDF), cumulative density function (CDF), and Laplace Transform of $X$, respectively. $\:_2F_1[a,b,c,z]$ represents Gauss Hyper-geometric function and $\mathrm{erf}(x)=1-\frac{2\int_x^\infty e^{-t^2} dt}{\sqrt{\pi}}$ represents error function.

\section{System Model}

\subsection{Network Deployment Model}
We consider a two-tier network with optical BSs (OBSs) and RF small-cell BSs (SBSs)  distributed according to a 2-D homogeneous PPP $\Phi_o$ and $\Phi_s$ with a density of $\lambda_o$ and $\lambda_s$, respectively. 
Note that the impact of
neighboring BSs that are located far away from the considered user
can be minimized by modeling the FOV limits of the PD receiver.
Each OBS consists of an LED lamp with
several LEDs. 
All OBSs reuse the same bandwidth so that
there is inter-cell interference (ICI). 
We consider that the LEDs are point sources with Lambertian emission and that they operate within the linear dynamic range of the current-to-power characteristic curve. It is also assumed that the LEDs are oriented vertically downwards. 
The users are distributed as a homogeneous PPP over
the coverage area $A$, i.e., the number
of users are Poisson distributed with intensity  $\lambda_u=c |A|$ where $c$ is the average number of users. The analysis is performed for a typical user  located at origin. 
Each user is equipped with both VLC and RF receivers. 

\subsection{Channel Model}

\subsubsection{RF Channel}
The RF communication channel power between a user and $i$-th SBS captures both channel fading and path-loss. The RF power is modeled using WINNER-II channel model as 
$
P^{\mathrm{rf}}_{i}= K v_i^{-\alpha} \chi_i,
$
where $\chi_i$ is the Nakagami fading channel (with Gamma distributed fading power with $\kappa$ and $\Theta$ as shape and scale parameters, respectively), $\alpha$ is the path-loss exponent, $v_i$ is the distance of the typical user to SBS, $K=10^{\frac{X}{10}}$, $X=B+C\mathrm{log}_{10}\left(\frac{f_c}{5}\right)$, $f_c$ 
is the carrier frequency in GHz, $B$ and $C$
are constants depending on the propagation model. For the LoS scenario, $B
=46.8$ and $C=20$.  For the non-LOS scenario, $B=43.8$ and $C=20$. 
Gamma distribution
is a versatile fading
distribution  which  includes  Rayleigh  distribution  for
$\kappa=1$
(for  non-LoS  conditions)  as  a  special  case
and  can  well-approximate  the  Rician  fading  distribution  for
$1 \leq \kappa \leq \infty$
(for strong LoS conditions).

\subsubsection{VLC Channel}
Each OBS is treated as a point source and the PD is installed on the user device facing upward.  The channel DC gain between a user and $i$-th OBS can be modeled using Lambertian emission model as follows~\cite{chmodel1,chmodel2}:
\begin{align} \label{vlcpower}
G^{\mathrm{vlc}}_{i}=&\frac{A_{\mathrm{pd}}(m+1)}{2 \pi u_{i}^2} \mathrm{cos}^m(\phi_i) T(\xi_i) G(\xi_i) \mathrm{cos}(\xi_i),
\end{align}
where $u_{i}$ is the distance of a typical user to $i$-th OBS, $m$ denotes the order of Lambertian emission and can be calculated as $m=-\frac{1}{\mathrm{log}_2(\mathrm{cos}(\Phi_{1/2}))}$, where $\Phi_{1/2}$ is the angle of radiance at which the emitted optical power from OBS is half of that emitted with $\Phi_{1/2}=0$.  $\phi_{i}$ represents the $i$-th LED/OBS irradiance angle with respect to the typical user, $\xi_i$ is the angle of incidence of $i$-th OBS to the typical user, $A_{\mathrm{pd}}$ denotes the detection area of the PD, $T(\xi_i)$ is the gain of the receiver's optical filter, and gain of the non-imaging concentrator can be given as $
G(\xi_i)=\frac{n^2}{\mathrm{sin}^2 \xi_{{\mathrm{fov}}}}, \:\:0\leq \xi_i \leq \xi_{\mathrm{fov}},
$
where $n$ is the ratio of {{speed of light in vacuum}} and {{velocity of light in the optical material}}, and $\xi_{\mathrm{fov}}$ is the half of the PD's FOV.  For visible light, the typical values of
$n$ lie between 1 and 2. 
\begin{figure*}
\centering
\includegraphics[scale=0.6]{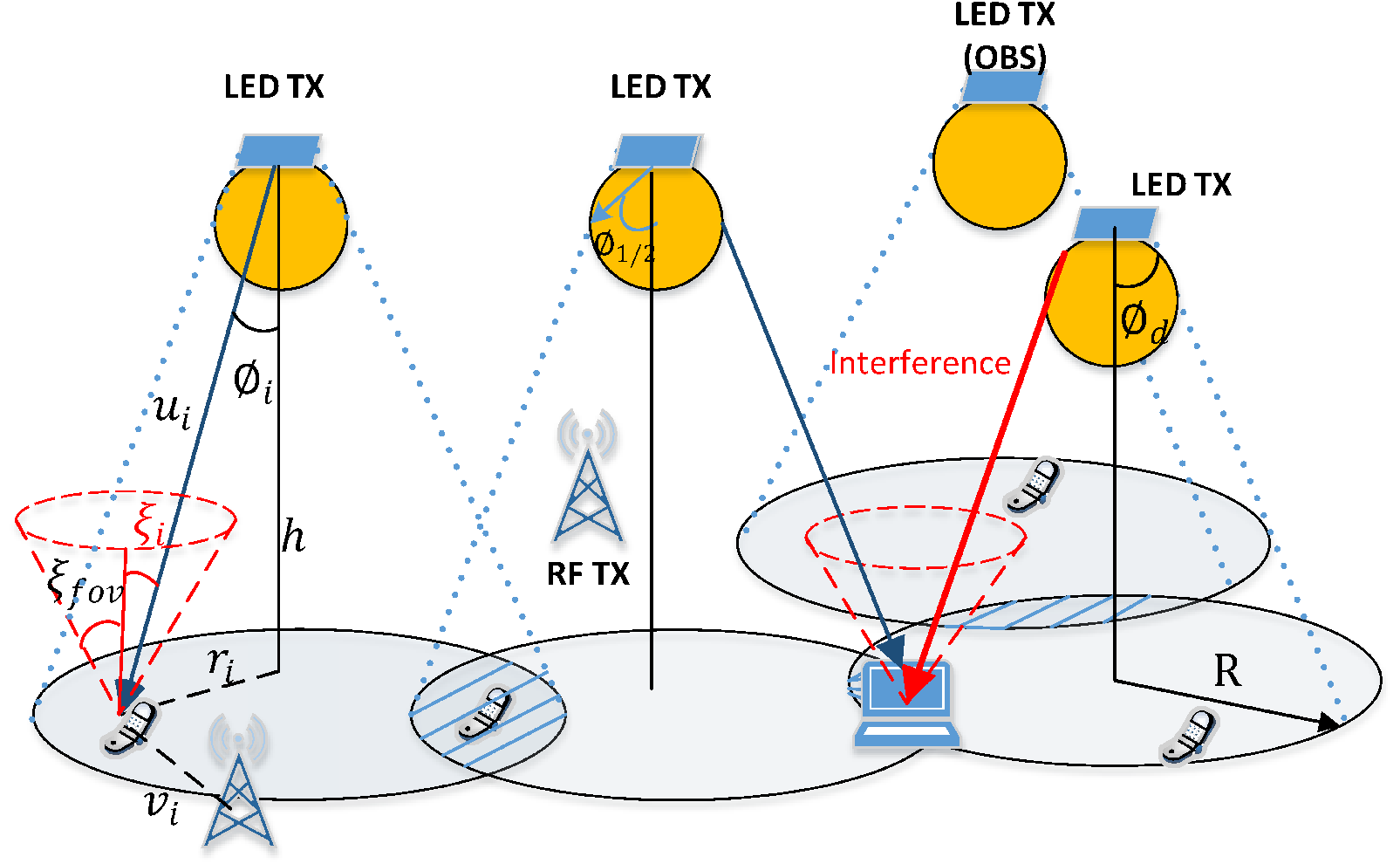}
\caption{Illustration of the downlink hybrid VLC/RF system which covers
a large indoor area.}
\label{vlc}
\end{figure*}
For LoS case\footnote{In this paper, reflection paths are not
considered  for
indoor visible light propagation. It is shown in \cite{chen34downlink,shao2015design} that the reflection paths have
an insignificant effect on the attocells that are sufficiently away
from the wall boundaries. Nonetheless, this is not a limitation and we will discuss an approach to incorporate reflections into the derivations in Section~V.} and given the geometrical illustration in \figref{vlc}, we can observe that $\mathrm{cos}\phi_i=\mathrm{cos}\xi_i=\frac{h}{u_{i}}=\frac{h}{\sqrt{r_i^2+h^2}}$, where $h$ denotes the fixed vertical separation between OBSs and the
user devices and  $r_i$ denotes the horizontal separation between the typical user and $i$-th OBS. \eqref{vlcpower} can then be rewritten as:
\begin{align}
G^{\mathrm{vlc}}_{i}
=\frac{A_{\mathrm{pd}}(m+1) T(\xi_i) G(\xi_i) h^{m+1}}{2 \pi (r_{i}^2+h^2)^{\frac{m+3}{2}}}, \: \: \:\:0\leq\xi_i\leq \xi_{\mathrm{fov}}.
\end{align}
The channel power can then be obtained as $P^{\mathrm{vlc}}_{i}=(G^{\mathrm{vlc}}_{i})^2$.

\subsection{Association and SINR Model}
We consider maximum received signal power-based association criterion which is equivalent to the nearest BS association criterion  for RF-only and VLC-only scenarios. For opportunistic RF/VLC, the typical user opportunistically selects the BS with maximum received signal power for transmission\footnote{By normalizing $P_i^{\mathrm{vlc}}$ with optical noise power, the association based on maximum received SNR can also be handled using this framework.}.  For hybrid RF/VLC networks, the typical user associates and transmits to both  RF and VLC networks. The maximum bandwidth allocated for a typical user in VLC and RF networks is $B_{o}$ and $B_{s}$, respectively. 

The SINR of a typical user  from its associated SBS (say $\mathcal{B}_0$) is given by

\begin{equation}
\gamma^{\mathrm{rf}}=\frac{P_s P_0^{\mathrm{rf}}}{\sum_{i \in \Phi_s \backslash \mathcal{B}_0 } P_s P_i^{\mathrm{rf}}+B_s N_s}, \nonumber
\end{equation}
where $P_i^{\mathrm{rf}}$ and $P_0^{\mathrm{rf}}$ represent
the received power of the typical user from $i$-th SBS and associated SBS, respectively, and the noise power spectral density is $N_{s}$.
The transmission power of a given  SBS is denoted by $P_{s}$.
For a typical user, the achievable data rate is given as ${B_s} \mathrm{log}_2(1+\gamma^{\mathrm{rf}})$. 

Similarly, the SINR of a typical user from  an OBS is given as 
$
\gamma^{\mathrm{vlc}}=\frac{ R^2_{\mathrm{pd}} P_o P_0^{\mathrm{vlc}}}{\sum_{i \in \Phi_o \backslash \mathcal{B}_0 } R^2_{\mathrm{pd}} P_o P_i^{\mathrm{vlc}}+B_o f^2 N_o},
$
where $P_i^{\mathrm{vlc}}$ and $P_0^{\mathrm{vlc}}$ represent
the received power of the typical user from $i$-th OBS and associated OBS, respectively, $R_{\mathrm{pd}}$ denotes the optical to electric conversion efficiency at the
receiver,  $N_{o}$ corresponds to the noise power spectral density caused by the received optical signal,  received ambient light (mainly daylight), and  thermal noise in the receiver circuit~\cite{chen34downlink}, and  $f$ denotes the ratio between the average transmitted optical
power (which is also proportional to the DC bias) and the
electrical power of the information signals without DC bias. 
Typically, as $f$ increases the probability of information signal
being outside the LED linear working region decreases. For
instance, $f$ = 3 means that approximately 0.3\% of the signal
is clipped. In this case, the clipping noise can be considered as
negligible. In DCO-OFDM, the achievable data rate by a
typical user can be expressed $\frac{B_o}{2} \mathrm{log}_2(1+\gamma^{\mathrm{vlc}})$~\cite{chen34downlink}.
The optical transmit power  allocated to a user by a given OBS is denoted by $P_{\mathrm{opt}}$. All LEDs fully reuse the  modulation bandwidth available and emit the same average optical power. 
Given the optical-to-electrical conversion efficiency $\kappa$, the electrical transmit power of an OBS is given as $P_{o}=\frac{P^2_{\mathrm{opt}}}{\kappa^2}$.
 

\section{Rate and Coverage Probability Analysis for Isolated RF and VLC Networks}
The coverage probability of a typical user is defined 
as the probability that its instantaneous SINR exceeds the target SINR threshold. Since the modulation bandwidth of a given OBS and SBS can vary, investigating the rate coverage probability of a typical user is more precise. As such, we define the target rate requirement of a typical user as $R_{\mathrm{th}}$ and, subsequently, the minimum target SINR of the typical user can be given as $\tilde\gamma^{\mathrm{vlc}}=2^{\frac{R_{\mathrm{th}}}{B_{\mathrm{vlc}}}}-1$ and $\tilde\gamma^{\mathrm{rf}}=2^{\frac{R_{\mathrm{th}}}{B_{\mathrm{rf}}}}-1$ for VLC and RF networks, respectively.

The  average  achievable  rate  is  another  important  performance  metric  in  a  wireless  communication  system to define the average data  rate  that  a  cellular  network  can  support on a given bandwidth.
The  average rate (or more precisely spectral efficiency) can be expressed as:
\begin{equation}\label{cap01}
\mathbb{E}[\mathrm{log}_2(1+\mathrm{SINR})]=\frac{1}{\mathrm{ln}(2)} \int_0^\infty \frac{\mathbb{P}(\mathrm{SINR}>t) }{t} dt.
\end{equation}

\subsection{Coverage Probability for an RF-Only Network}
The coverage probability of a typical user $(\mathcal{C}_{s})$ in an RF-only environment considering  Nakagami-$m$ fading (i.e., $\chi\sim \mathrm{Gamma}(\kappa, \Theta)$) where $\kappa$ and $\Theta$ represents shadowing severity and scaling parameters, respectively,  can be derived 
as follows:
\begin{align}
\mathcal{C}_{s}=&\mathbb{E}_v\left[\mathbb{P}\left(\frac{K P_s v^{-\alpha} \chi_0}{I+B_s N_s} >\tilde\gamma^{\mathrm{rf}}\right)\right]
\stackrel{(a)}{=}
\mathbb{E}_v\left[\mathbb{E}_I\left[\frac{\Gamma_u(\kappa, \frac{\tilde\gamma^{\mathrm{rf}}(I+B_s N_s)}{K P_s v^{-\alpha} \Theta})}{\Gamma(\kappa)}\right]\right],
\nonumber\\
\stackrel{(b)}{=}&
\mathbb{E}_{v,I}\left[\mathrm{exp}\left(- \frac{\tilde\gamma^{\mathrm{rf}}(I+B_s N_s)}{K P_s v^{-\alpha} \Theta}\right) \sum_{n=0}^{\kappa-1} \frac{(\frac{\tilde\gamma^{\mathrm{rf}}(I+B_s N_s)}{K P_s v^{-\alpha} \Theta})^n}{n!}\right],
\nonumber\\\stackrel{(c)}{=}&
\mathbb{E}_{v}\left[\sum_{n=0}^{\kappa-1} 
\left.\frac{(-s)^n}{n!}
\frac{d^n}{ds^n} \mathcal{L}_{Z} (s)\right]\right|_{s=\frac{\tilde\gamma^{\mathrm{rf}} v^{\alpha}}{K P_s  \Theta}},
\label{cov1}
\end{align}
where $I=\sum_{i \in \Phi_s \backslash B_0 } K P_s \tilde{v}_{i}^{-\alpha} \chi_i$ and $f_V(v)=\frac{2 \pi \lambda_s v e^{-\pi \lambda_s v^2}}{U}$ is truncated exponential distribution, where ${U}=1-e^{-\pi \lambda_s R_m^2}$ is the probability of at least one SBS with in $R_m$.   Since $R_m$ is large, $U \approx 1$.
The  value $\kappa=1$ results  in  the  Rayleigh-fading model, whereas the values of $\kappa <1$ represent 
channel fading more severe than  Rayleigh  fading  and  values
of $\kappa > 1$ correspond  to  channel  fading  less  severe  than Rayleigh fading. Note that (a) follows from the definition of the CDF of Gamma distribution, (b) follows from the definition of the upper incomplete Gamma function for integer parameters $\frac{\Gamma(\kappa, x)}{\Gamma(\kappa )}=\mathrm{exp}(- x) \sum_{n=0}^{\kappa-1} \frac{x^n}{n!}$, and (c) follows from the Laplace Transform property $\int_0^\infty e^{-s x} x^n f(x) dx=(-1)^n \frac{d^n}{ds^n} \mathcal{L}_{X} (s)$ and considering $Z=I+B_s N_s$. Note that  $
\mathcal{L}_Z(s)=\mathcal{L}_I(s) e^{-s  B_s N_s}$ and a closed-form for $\mathcal{L}_{I}(s)$ can be derived as:
\begin{align}
\mathcal{L}_{I}(s)=&\mathrm{exp}\left(-2 \pi \lambda_s \int_r^{R_m} \frac{(1 + { s  \tilde{v}^{-\alpha}})^\kappa-1}{(1 + { s \tilde{v}^{-\alpha}})^\kappa} \tilde{v} d\tilde{v}\right)
=
\mathrm{exp}(-\pi \lambda_s (\Phi(s,R_{m})-\Phi(s,r))),
\label{LI}
\end{align}
where $\Phi(s,x)=x^2-x^2\:_2F_1[\kappa,-\frac{2}{\alpha},1-\frac{2}{\alpha},-s P_s x^{-\alpha}]$. 

Although the aforementioned approach provides an exact evaluation of the coverage probability, its computation can be cumbersome depending on the value of $\kappa$. The higher values of $\kappa$ require higher-order derivatives of the Laplace Transform. Motivated by the aforementioned complexity issues, in the following, we discuss a more tractable approximation of the coverage probability for Nakagami-$m$ fading channels using the bounds on incomplete Gamma function as given  in the following~\cite{gammaapprox}:
\begin{equation}\label{approx}
(1-e^{-p x})^\kappa < \frac{\Gamma_l(\kappa,x)}{\Gamma(\kappa)}<(1-e^{- \tilde{p} x} )^\kappa, 
\end{equation}
where $0 \leq x\leq \infty$, $\kappa>0$, and $\kappa \neq 1$, where
\begin{equation}
\tilde p=
\begin{cases}
(\kappa!)^{-\frac{1}{\kappa}}, & 0 <\kappa<1\\
1, & \kappa>1\\
\end{cases},
\quad\quad
 p=
\begin{cases}
1, & 0 <\kappa<1\\
(\kappa!)^{-\frac{1}{\kappa}}, & \kappa>1\\
\end{cases}.
\end{equation}
As $\kappa \rightarrow 1$, both bounds tend to the exact value $1-\mathrm{exp}(-x)$, which is the value of $\gamma(1,x)$. In the domain $0 < \kappa <1$, the bounds (also referred to as {\em Alzer's inequality} in \cite{gammaapprox,bai2015coverage})  are sharper if $x$ is small. Using the aforementioned inequality, the coverage probability can be approximated as in the following.

\begin{lemma}[Approximate Coverage Probability in Nakagami-$m$ fading Channels] To approximate the coverage probability, we use the tight upper bound of the CDF of the Gamma RV from \eqref{approx} as
$
\mathbb{P}(\chi_0 \leq T)=\frac{\Gamma_l[\kappa, \frac{T}{\Theta}]}{\Gamma(\kappa)}< (1-e^{-\tilde p T })^\kappa,
$
where $\tilde p=(\kappa!)^{-1/\kappa}$ and $T>0$. This bound approximates the tail probability of a gamma RV. 
The approximate coverage probability can then be given as:
\begin{equation*}
1- \sum_{l=1}^{\kappa} (-1)^l {\kappa \choose l} \mathbb{E}_v\left[ e^{- \frac{p l \tilde\gamma^{\mathrm{rf}} B_s N_s}{P_s v^{-\alpha}}} \mathcal{L}_{I}\left(\frac{p l \tilde\gamma^{\mathrm{rf}}}{P_s v^{-\alpha}}\right)\right].
\end{equation*}
\end{lemma}
\begin{proof}
From \eqref{cov1}, we can write the coverage probability as follows:
\begin{align*}
\mathcal{C}_{s}=&
\mathbb{E}_v\left[\mathbb{E}_I\left[1-\frac{\Gamma_l(\kappa, \frac{\tilde\gamma^{\mathrm{rf}}(I+B_s N_s)}{K P_s v^{-\alpha} \Theta})}{\Gamma(\kappa)}\right]\right]
\stackrel{(a)}{\approx}
1-
\mathbb{E}_v\left[\mathbb{E}_I\left[
\left(1-e^{-\tilde p \frac{\tilde\gamma^{\mathrm{rf}}(I+B_s N_s)}{K P_s v^{-\alpha} \Theta})}\right)^\kappa
\right]\right],
\nonumber\\
\stackrel{(b)}{=}&
1-
\mathbb{E}_v\left[\mathbb{E}_I\left[
\sum_{l=0}^{\kappa} {\kappa \choose l} (-1)^l 
e^{-p l \frac{\tilde\gamma^{\mathrm{rf}}(I+B_s N_s)}{K P_s v^{-\alpha} \Theta})}
\right]\right],
\end{align*}
where (a) is obtained using the lower bound and (b) is obtained using the Binomial expansion. Applying the definition of Laplace Transform $\mathcal{L}_{I}(\cdot)$, the coverage expression can be given as in {\bf Lemma~1}, where the closed-form expression for $\mathcal{L}_{I}(\cdot)$ can be given as in \eqref{LI}.
\end{proof}

For the Rayleigh fading,  exact coverage probability can be derived as follows:
\begin{align}\label{rfcov}
\mathcal{C}_s= \int_0^{R_m} \frac{2 \pi \lambda_s}{U}  
e^{-\pi \lambda_s v^2(1+\rho)-\frac{\tilde\gamma^{\mathrm{rf}} B_s N_s}{P_s v^{-\alpha}}} v dv,
\end{align}
where $\tilde{v}$  denotes the distance between typical user and interferers and 
$
\rho(r)=\int_{(\tilde\gamma^{\mathrm{rf}})^{-\frac{2}{\alpha}}}^{\frac{R_m^2 (\tilde\gamma^{\mathrm{rf}})^{-\frac{2}{\alpha}}}{r^2}} \frac{(\tilde\gamma^{\mathrm{rf}})^{2/\alpha}}{1+u^{\alpha/2}} du.
$
Note that the general approach to derive the coverage probability remains the same as in~\cite[Theorem~1]{andrews2011tractable}. However, due to the finiteness of the considered indoor area, $\rho$ depends on $r$, thus leading to a finite integral and a different closed-form solution for $\rho(r)$ as 
\begin{equation}
\frac{v^2 \:_2F_1\left[1,\frac{2}{\alpha}, 1+\frac{2}{\alpha}, \frac{-1}{\tilde\gamma^{\mathrm{rf}}}\right] -R_{m}^2\:_2F_1[1,\frac{2}{\alpha}, 1+\frac{2}{\alpha}, -\frac{(\frac{R_m}{r})^\alpha}{\tilde\gamma^{\mathrm{rf}}}] }{(\tilde\gamma^{\mathrm{rf}})^{\frac{2}{\alpha}} v^2}.
\end{equation}

\subsection{Coverage Probability for a VLC-Only Network}

Since we assume maximum received signal power (nearest BS)-based association for VLC-only networks, the distribution of the distance of the typical user with its associated BS  can be given as 
$
f_{r}(r)={2 \pi \lambda_o r \mathrm{exp}(-\pi \lambda_o r^2)}/U.
$
The coverage probability of a typical user relies on the probability that at least an OBS  should be located within the FOV of the typical user. Consequently, for a given distance $u=\sqrt{h^2+r^2}$ between the typical user and an OBS, the probability can be given as in the following.

\begin{lemma}
An OBS exists within the FOV of the typical user iff 
$
\mathrm{tan}^{-1}\left(\frac{h}{r}\right)=\mathrm{cos}^{-1}\left(\frac{h}{u} \right)\leq \xi_{\mathrm{fov}}
$. That is,
$
r \leq h\mathrm{tan}(\xi_{\mathrm{fov}}) =\hat{\mathcal{T}}
$
given that $\hat{\mathcal{T}} \leq R_m$. 
Consequently, defining $\mathcal{T} =\mathrm{min}(R_m, \hat {\mathcal{T}})$, we can write the probability of at least one BS within FOV of the typical user as follows.
\begin{align}\label{radi}
\mathbb{P}(r \leq {\mathcal{T}})=[{1-\mathrm{exp}(-\lambda_o \pi \mathcal{T}^2)}]/U.
\end{align}
\end{lemma}
\begin{proof}
Using the null probability of a 2-D PPP 
$\Phi_o$,  $\mathbb{P}(r \leq \hat{\mathcal{T}})={1-\mathrm{exp}(-\lambda_o \pi \hat{\mathcal{T}}^2)}$ whereas, if $\hat{\mathcal{T}} \geq R_m$, $\mathbb{P}(r \leq \hat{\mathcal{T}})=\mathbb{P}(r \leq R_m)=1$. Subsequently, $\mathbb{P}(r \leq {\mathcal{T}})$ can be given as in {\bf Lemma~2}. 
\end{proof}

As the number of interferers for a small indoor area and/or smaller FOV  of the PD receiver is quite small, the coverage probability  in noise-limited regime is of practical relevance.

\begin{lemma}[Coverage Probability in Noise-Limited Regime] 
The coverage probability of a typical user in a VLC network can be given in closed-form  for noise-limited regime as follows:
\begin{align*}
\mathcal{C}&=\mathbb{P}\left( 
\frac{R^2_{\mathrm{pd}} P_o Z}{\tilde\gamma^{\mathrm{vlc}} {{B}_o f^2  N_o}}
>(r^2+h^2)^{m+3}\right),
\\\nonumber
&
\stackrel{(a)}{=}
\frac{1}{U}
\begin{cases}
{1-e^{-\pi \lambda_o \left(\left(\frac{R^2_{\mathrm{pd}} P_o Z}{\tilde\gamma^{\mathrm{vlc}} {B}_o f^2 N_o}\right)^{\frac{1}{m+3}}-h^2\right)}}, \; \mathrm{if}
& \left(\frac{R^2_{\mathrm{pd}} P_o Z}{\tilde\gamma^{\mathrm{vlc}} {B}_o f^2 N_o}\right)^{\frac{1}{m+3}}-h^2 \leq \mathcal{T}^2\\
1-e^{-\pi \lambda_o \mathcal{T}^2},& \mathrm{otherwise}
\end{cases}.
\end{align*}
Note that the coverage occurs only when $r \leq \mathcal{T}$ otherwise there is no coverage. 
\end{lemma}

\begin{lemma}
The exact coverage probability of a typical user $\mathcal{C}_{o}$ can then be derived as follows:
\begin{align*}
\mathcal{C}_{o}=\left(\frac{1}{2}-\frac{1}{\pi} \int_0^\infty 
\frac{\mathrm{Im}[\phi_{\Omega}(\omega)]}{\omega} e^{- j \tilde\gamma^{\mathrm{vlc}} {B}_o f^2 N_o \omega} d\omega \right)\mathbb{P}(r \leq \mathcal{T}),
\end{align*}
where $\Omega=X-\tilde\gamma^{\mathrm{vlc}} \mathcal{I}_a$,
$\mathcal{I}_a=  \sum_{i \in \Phi_a \backslash B_0 } R^2_{\mathrm{pd}} P_o Z (\tilde{r}_{i}^2+h^2)^{-(m+3)} 
$, $\tilde{r}$ represents the distance between typical user and interferers,  $Z=\left(\frac{A_{\mathrm{pd}}(m+1) T(\xi) G(\xi) h^{m+1}}{2 \pi}\right)^2$, $\phi_\Omega(\omega) = \mathbb{E}[e^{-j \omega \Omega }]$ is the characterestic function (CF) of $\Omega$, and $\mathrm{Im}(\cdot)$ is the imaginary part of $\phi_\Omega(\cdot)$. 
\end{lemma}
\begin{proof}
The conditional coverage probability of a typical user $\mathcal{C}_{r\leq\mathcal{T}}$ can  be derived as:
\begin{align*}
\mathcal{C}_{r\leq\mathcal{T}}&=\mathbb{P}\left(\frac{ R^2_{\mathrm{pd}} P_o Z 
(r^2+h^2)^{-m-3}}
{\sum_{i \in \Phi_o \backslash B_0 }  R^2_{\mathrm{pd}} P_o Z 
(\tilde{r}_{i}^2+h^2)^{-m-3}+{B}_o f^2  N_o}>\tilde\gamma^{\mathrm{vlc}}\right),
\\\nonumber
&
=\mathbb{P}\left( \underbrace{R^2_{\mathrm{pd}} P_o Z 
(r^2+h^2)^{-m-3}}_{X} - \tilde\gamma^{\mathrm{vlc}} \mathcal{I}_a
>\tilde\gamma^{\mathrm{vlc}} {B}_o f^2  N_o)\right),
\\\nonumber
&
\stackrel{(a)}{=}\frac{1}{2}-\frac{1}{\pi} \int_0^\infty 
\frac{\mathrm{Im}[\phi_{\Omega}(\omega)]}{\omega} e^{- j \tilde\gamma^{\mathrm{vlc}} {B}_o f^2  N_o \omega} d\omega.
\end{align*}
Note that (a) is obtained from applying Gil-Pelaez inversion theorem. The exact coverage probability in {\bf Lemma~4} can then be derived by multiplying the conditional coverage probability $\mathcal{C}_{r\leq\mathcal{T}}$ with the probability $\mathbb{P}(r \leq \mathcal{T})$ derived in {\bf Lemma~2}.
\end{proof}

To derive $\phi_\Omega(\omega)$ and $\mathcal{C}_o$, our analytical methodology is summarized herein:
\begin{enumerate}
\item Derive the distribution of the number of interferers within the FOV of the typical user.
\item Derive the conditional Laplace transform and CF of the aggregate interference ($\mathcal{I}_a$) incurred at typical user. Since $X$ and $\mathcal{I}_a$ are dependent on $r$, we compute  the  conditional CF  $\phi_{\Omega|r}(\cdot)$.
\item Determine $\phi_{\Omega}(\cdot)$ defined  as follows:
\begin{align}\label{phi}
\phi_{\Omega}(\omega)=
\mathbb{E}_{r}[\phi_{\Omega|r}(\omega)]
=\mathbb{E}_{r}\left[e^{-j \omega X} \mathcal{L}_{\mathcal{I}_a|r}({-j \omega \tilde\gamma^{\mathrm{vlc}}})\right].
\end{align}
\item Derive the coverage probability as detailed in {\bf Lemma~4}.
\end{enumerate}
The details follow in the following subsections:

\subsubsection{Distribution of the Number of Interferers}
All potential optical interferers for a typical user will be located with in the region between $r$ and $\mathcal{T}$. As such, given the Poisson distribution, the probability mass function (PMF) of the number of  interferers $N$ located within $r$ and $\mathcal{T}$ can be given as follows:
\begin{equation}\label{no}
\mathbb{P}(N=k)=\frac{e^{-\lambda_o \pi (\mathcal{T}^2 -r^2)} (\lambda_o \pi (\mathcal{T}^2-r^2))^k }{k!}.
\end{equation}

\subsubsection{Laplace Transform of Aggregate Interference}
Since the OBSs follow a homogeneous PPP, all $N$ interferers are independent and identically distributed. The Laplace Transform of $\mathcal{I}_a$, conditioned on $r$, can then be given as follows:
\begin{align}\label{Ia}
\nonumber
&\mathcal{L}_{\mathcal{I}_a|r}(s)= \sum_{k=0}^\infty \left(\mathbb{E}[\mathrm{exp}(-s \mathcal{I})|N]\right)^k \mathbb{P}(N=k)
\stackrel{(a)}{=}
\mathrm{exp}\left[\lambda_o \pi (\mathcal{T}^2-r^2)\left( \mathbb{E}[e^{-s \mathcal{I}}|N] -1\right)\right],
\nonumber\\
&
\stackrel{(b)}{=}
\mathrm{exp}\left[\lambda_o \pi (\mathcal{T}^2-r^2)\left( \mathbb{E}[e^{-s R^2_{\mathrm{pd}} P_o Z  (
\tilde{u})^{-2(m+3)}}]  -1\right)\right], 
\end{align}
where $\mathcal{I}=  R^2_{\mathrm{pd}} P_o Z (\tilde{r}^2+h^2)^{-(m+3)} $, (a) follows from the definition of the exponential function, and (b) follows by defining $\tilde{u} =\sqrt{\tilde{r}^2+h^2}$. 
Conditioned on $r$, the PDF of $\tilde{r}$ and $\tilde{u}$ can be given, respectively, as:
\begin{align*}
&f_{\tilde{r}}(\tilde{r})=\frac{2 \tilde{r}}{\mathcal{T}^2-r^2}, \quad r \leq \tilde{r} \leq \mathcal{T},
\\&
f_{\tilde{u}}(\tilde{u})=\frac{2 \tilde{u}}{\mathcal{T}^2-r^2}, \quad \sqrt{r^2+h^2} \leq \tilde{u} \leq \sqrt{\mathcal{T}^2+h^2}.
\end{align*}

Note that $0 \leq r \leq \mathcal{T} $ and $ h \leq u \leq \sqrt{\mathcal{T}^2+h^2}$. Therefore, $r \leq \tilde{r} \leq \mathcal{T} $ and $ \sqrt{r^2+h^2}\leq \tilde{u} \leq \sqrt{\mathcal{T}^2+h^2}$.
Applying the aforementioned distance distributions, a closed-form expression for $\mathcal{L}_{\mathcal{I}|r}(s)=\mathbb{E}[e^{-s R^2_{\mathrm{pd}} P_o Z   (\tilde{u})^{-2(m+3)}}]$ can be given as in the following.
\begin{lemma}[Conditional Laplace Transform of the Interference] To derive a closed form for the Laplace Transform of aggregate interference in \eqref{Ia}, we first derive $\mathcal{L}_{\mathcal{I}|r}(s)$ as follows:
\begin{align*}
\mathcal{L}_{I|r}(s)&=
\frac{2 \int_{\sqrt{r^2+h^2}}^{\sqrt{\mathcal{T}^2+h^2}} \tilde{u} e^{-s R^2_{\mathrm{pd}} P_o Z  (\tilde{u})^{-2(m+3)}} d\tilde{u}}{\mathcal{T}^2 -r^2}
=
\frac{ 
\Gamma\left(-\frac{1}{3+m}, 
\frac{Z R^2_{\mathrm{pd}} P_o s}{({h^2+\mathcal{T}^2})^{(m+3)}}, 
\frac{ZR^2_{\mathrm{pd}} P_o s}{({h^2+r^2})^{(m+3)}}\right)
}{(Z R^2_{\mathrm{pd}} P_o s)^{-\frac{1}{m+3}} (m+3)(\mathcal{T}^2-r^2)},
\end{align*}
where $\Gamma(z,x,y)=\Gamma_u(z,x)-\Gamma_u(z,y)$ is the Generalized incomplete Gamma function. 
\end{lemma}
Substituting the closed-form result of $\mathcal{L}_{I|r}(s)$ in \eqref{Ia}, we obtain the closed form expression for $\mathcal{L}_{\mathcal{I}_a}(s)$ which can be  substituted in \eqref{phi} to obtain the CF of $\Omega$. After averaging over $r$, we can obtain the CF of $\Omega$   and ultimately the coverage probability using {\bf Lemma~4}. The closed-form expression for $\mathcal{L}_{I|r}(s)$ is a function of the incomplete Gamma function which is a built-in function in standard mathematical software packages such as \texttt{MATLAB} and \texttt{MATHEMATICA}. Note that due to the negative sign of the first argument of upper incomplete Gamma function, the {\em Alzer's inequality} cannot be applied.

{\bf Remark:} A   simplified  expression for $\mathcal{L}_{I|r}(s)$ can be derived using the asymptotic approximation of the incomplete Gamma function for small values of $x$ and $s$ as given below:
\begin{equation}
\lim_{x\rightarrow 0}\Gamma_u(z,x) \approx -\frac{x^z}{z}.
\end{equation}
This approximation is of interest in  our specific application as the arguments of incomplete Gamma function exist in the neighborhood of zero especially due to the factors $({h^2+r^2})^{-(m+3)}$ and ${({h^2+\mathcal{T}^2})^{-(m+3)}}$ where $m \geq 0$. The error  for real values of $z$ is on the order of $\mathcal{O}(x^{\min(z + 1, 0)})$ if $z \neq -1$ and $\mathcal{O}(\mathrm{ln}(x))$ if $z = -1$. In such scenarios, $\phi_\Omega (\omega)$ can be simplified as follows:
\begin{eqnarray}
\phi_\Omega (\omega)& = & \sum_{k=0}^{\infty} \frac{\pi \lambda_o e^{\lambda_o \pi h^2} (-1)^{\frac{(7+2m)(k+1)}{4+m}}}{(m+4)(k+1)! M^{\frac{1-mk-3k}{m+4}}} \nonumber \\
 &  & \quad \quad \Scale[1.1]{\left({\frac{\:_2F_1[1+k, \frac{(1+k)(3+m)}{m+4},k+2,g]}{g^{-1-k}} 
- \frac{\:_2F_1[1+k, \frac{(1+k)(3+m)}{m+4},k+2,n]}{n^{-1-k}} }\right)}, \nonumber
\end{eqnarray}
where $g=M h^{2(m+4)}$, $n=M (\mathcal{T}^2+h)^{(m+4)}$, and $M=\frac{\lambda_o \pi}{Z s_0}$.

\subsection{Average Achievable Spectral Efficiency}

As we can observe that the  approach for the ergodic capacity calculation in \eqref{cap01} requires another integral on top of the coverage probability expression; therefore, the approach can be computationally intensive. Instead, we resort to another approach where the spectral efficiency can be derived using the Laplace Transform of the received signal and interference power~\cite{hamdi2010useful}: 
\begin{equation}
\mathbb{E}\left[\mathrm{ln}\left(1+\frac{X}{1+Y}\right)\right]=\int_0^\infty \frac{\mathcal{L}_{Y}(s)-\mathcal{L}_{X,Y}(s)}{s}{e^{-s}} ds,
\end{equation}
where $\mathcal{L}_{Y}(s)=\mathbb{E}[e^{-s Y}]$ and $\mathcal{L}_{X,Y}(s)=\mathbb{E}[e^{-s (X+Y)}]$ is the joint MGF of $X$ and $Y$.
Now we can derive the spectral efficiency of a typical user using Hamdi's lemma for RF-only networks as follows:
\begin{equation}\label{Rs}
\mathcal{R}_s=\int_0^\infty \frac{\mathbb{E}_v\left[\mathcal{L}_{{I}|v}(s)-\mathcal{L}_{I,X|v}(s)\right]}{s}{e^{-{B}_s N_s}} ds,
\end{equation}
where $\mathcal{L}_{{I}|v}(s)$ can be given as in \eqref{Ia} and $X= K P_s v^{-\alpha} \chi_0$.
Conditioned on $v$, the variables ${I}$ and $X$ are independent, therefore  $
\mathcal{L}_{I,X|v}(s)=e^{-s X} \mathcal{L}_{I|v}({s}),
$
where  $\mathcal{L}_{X|v} (s)=(1-s \Theta P_s v^{-\alpha})^{-\kappa}$.

Subsequently, we can derive the spectral efficiency using Hamdi's lemma for VLC-only networks as:
\begin{equation}\label{Ro}
\mathcal{R}_o=\int_0^\infty \frac{\mathbb{E}_r[\mathcal{L}_{\mathcal{I}_a|r}(s)-\mathcal{L}_{\mathcal{I}_a,X|r}(s)]}{s}{e^{-{B}_o f^2  N_o s}} ds,
\end{equation}
where $\mathcal{L}_{\mathcal{I}_a|r}(s)$ can be given as in \eqref{LI}. Note that conditioned on $r$, the variables $\mathcal{I}_a$ and $X$ are independent, therefore $\mathcal{L}_{\mathcal{I}_a,X}(s)$ can be given as $
\mathcal{L}_{\mathcal{I}_a,X|r}(s)=e^{-s X} \mathcal{L}_{\mathcal{I}_a|r}({s}).
$

\section{Rate and Coverage Probability Analysis in Co-existing RF/VLC Networks}

In this section, we derive the coverage probability  of a typical user in two scenarios, (i) {\em Hybrid RF/VLC:} a co-existing RF/VLC network where a user can communicate to both RF and VLC networks at the same time  and (ii) {\em Opportunistic RF/VLC:} a co-existing RF/VLC network where a user opportunistically selects the network with maximum received signal power.

\subsection{Coverage With Hybrid RF/VLC communication}
In a hybrid RF/VLC network, the rate outage will happen if and only if both RF and VLC link rate to the typical user goes below the target rate threshold $R_{\mathrm{th}}$. Subsequently, given the rate coverage probabilities of a typical user in  RF-only and VLC-only networks as in \eqref{rfcov} and {\bf Lemma~4}, respectively, we can derive the coverage probability of a typical user with hybrid RF/VLC communications as
$
\mathcal{C}_{h}= 1-(1-\mathcal{C}_o)(1-\mathcal{C}_s).
$
This network can also be considered as a VLC prioritized network where a user may first try to associate to a VLC network and switches to an RF network only in case of rate outage.  In such a case, an outage will occur only if user will not get rate coverage from any of the RF or VLC network.

\subsection{Coverage With Opportunistic RF/VLC communication}

We  first derive the association probability of a typical user with RF and VLC networks and then characterize the coverage probability of the typical user.
Since the typical user can associate to either network of OBSs or SBSs, from the law of total probability, the coverage probability can be given as:
\begin{equation}
\mathcal{C}=\mathbb{P}_o \tilde{\mathcal{C}}_{o} +(1-\mathbb{P}_o) \tilde{\mathcal{C}}_{s},
\end{equation}
where $\mathbb{P}_o$ is the  association probability, $\tilde{\mathcal{C}}_{o}$ and $\tilde{\mathcal{C}}_{s}$ are the coverage probabilities of
a typical user when associated with the tier of OBSs and SBSs, respectively. Note that $\tilde{\mathcal{C}}_{o}$ and $\tilde{\mathcal{C}}_{s}$ are different than ${\mathcal{C}}_{o}$ and ${\mathcal{C}}_{s}$ since the opportunistic selection of the network changes the distribution of the distance between user and the associated OBS or SBS. The details will follow in the subsequent subsections.

\subsubsection{Association Probability}

Given the maximum received signal power criterion, the association probability of a typical user to the OBS $\mathbb{P}_o$ depends on the probability of two events: (i) the probability that the received signal power from an OBS is higher than the received signal power from all SBSs, (ii)~the probability that OBS  is located within the FOV of the typical user. As such, the association probability of a typical user can be given as in the following.
\begin{lemma}[Association Probability of a Typical User to OBS] Conditioned on \eqref{radi}, the probability that a typical user associates to the OBS, $\mathbb{P}_o$ can be derived as follows:
\begin{equation}
\mathbb{P}_o=\frac{2 \pi \lambda_o}{U}\int_{0}^{\mathcal{T}} e^{-\lambda_s \pi Z_1 \left(r^2+h^2 \right)^{\frac{2(m+3)}{\alpha}} -\lambda_o \pi r^2} r  dr,
\end{equation}
where $Z_1=\left(\frac{ P_s K}{Z R^2_{\mathrm{pd}}P_o}\right)^{\frac{2}{\alpha}}$ and $Z=\left(\frac{A(m+1) T(\xi) G(\xi) h^{m+1}}{2 \pi}\right)^2$.
\end{lemma}
\begin{proof}
A typical user will associate to either OBS or SBS depending on the maximum received signal power. As such, the association probability with OBS can be given as:
\begin{align}
\mathbb{P}_o=& \mathbb{E}_{r}\left[\mathbb{P}\left[P^{\mathrm{vlc}} (r) >  P^{\mathrm{rf}} (v)\right]\right] 
\nonumber\\
=& \mathbb{E}_{r}\left[\mathbb{P}\left[{v}  > \left(\frac{P_s K  (r^2+h^2)^{{m+3}}}{Z R^2_{\mathrm{pd}}P_o} \right)^{\frac{1}{\alpha}}\right]\right],
\end{align}
where $Z=\left(\frac{A(m+1) T(\xi) G(\xi) h^{m+1}}{2 \pi}\right)^2$. Using the  null probability of a 2-D Poisson process with density $\lambda_s$ in an area $A$, $\mathbb{P}(v \geq x)= \mathbb{P}[\mathrm{No\:BS\: closer\: than\: x \:in \:SBS\: tier}]$ can be derived as 
$\mathbb{P}(v \geq x)= \mathrm{exp}(-\lambda_s |A|)=\mathrm{exp}(-\lambda_s \pi x^2)
$.
Subsequently,
$\mathbb{P}_o$ can be derived as: 
\begin{align}
\mathbb{P}_o=& \mathbb{E}_{r}\left[\mathrm{exp}\left(-\lambda_s \pi \left(\frac{P_s K  (r^2+h^2)^{{m+3}}}{Z R^2_{\mathrm{pd}}P_o} \right)^{\frac{2}{\alpha}}\right)\right],
\end{align}
Averaging over  the  PDF of $r$ given   as $
f_{r}(r)={2 \pi \lambda_o r \mathrm{exp}(-\lambda_o \pi r^2)}/U
$, {\bf Lemma~6} can be obtained. 
\end{proof}
The integral in {\bf Lemma~6} can be solved in closed-form for certain cases as shown below. 
\begin{corollary}[Closed Form Association Probability with OBS, $m+3=\alpha$] Given $m+3=\alpha$, the closed form association probability of a typical user with OBS can be derived as:
\begin{equation}
\mathbb{P}_o= 
\mathrm{exp}\left({h^2 \pi \lambda_o+\frac{Z_2^2}{\pi}}\right)
Z_2
\left[\mathrm{erf}\left(\frac{Z_2 (1+Z_3 h^2)}{\sqrt{\pi}} \right)-
\mathrm{erf}\left(\frac{Z_2 (1+Z_3 h^2(1+\mathrm{tan}^2(\xi_{\mathrm{fov}})))}{\sqrt{\pi}} \right)\right], \nonumber
\end{equation}
where
$Z_2=\frac{\pi  \lambda_o}{2 \sqrt{Z_1 \lambda_m}}$, $Z_3=\frac{2 \lambda_m Z_1}{\lambda_o}$, and $\xi_{\mathrm{fov}} \leq 90^\circ$.
\end{corollary}

{\bf Remark:} In order to prioritize VLC over RF network, the association to VLC networks can be maximized using the aforementioned expressions and optimal system parameters can be determined numerically on standard mathematical software packages such as \texttt{MAPLE} and \texttt{MATHEMATICA}.

The association probability can be further simplified using  an approximation for $\mathrm{erf}(\cdot)$ as follows:
$
\mathrm{erf}{(x)}\approx \sqrt{1-\mathrm{exp}\left(-x^2 \frac{\frac{4}{\pi} +a x^2}{1+ a x^2}\right)},
$
where $a=\frac{8 (\pi -3)}{3 \pi (4-\pi)}\approx 0.140012$.
This approximation is used since it is designed to be very accurate in a neighborhood of 0 and a neighborhood of infinity, and the error is less than 0.00035 for all $x$. Using the alternate value $a=0.147$ the maximum error reduces to about 0.00012.

The association probability to a certain network also provide direct insights into the  mean traffic load associated to each network. For instance, the mean traffic load at RF and VLC network can  be given as $\lambda_u(1-\mathbb{P}_{o})$ and $\lambda_u \mathbb{P}_{o}$, respectively. In order to distribute the traffic load in RF and VLC network according to the choice of network operator, we can select the system parameters (such as intensity of OBSs $\lambda_o$, intensity of SBSs $\lambda_s$, received power from OBSs) that satisfy $\mathbb{P}_o=\beta$, where $0 \leq \beta \leq 1$ is the proportion of users in VLC network. By numerically solving $\mathbb{P}_o=\beta$, the desired system parameters can be obtained.

\begin{corollary}[Closed-Form Solution for System Parameters for Required Traffic Offloading]  Using a more tractable approximation of the $\mathrm{erf}(\cdot)$, i.e.,
\begin{equation}\label{approx2}
\mathrm{erf}(x) \approx \frac{1}{\mathrm{exp}(x^2) \sqrt{x^2}},
\end{equation}
in Corollary~1 and considering $\xi_{\mathrm{fov}}=90^\circ$ closed-form solution for $Z_1$, $\lambda_o$ and $\lambda_s$ for required traffic offloading can be determined, respectively, as follows:
\begin{equation*}
Z_1=\frac{-\lambda_o \pi h^2 + 2 \mathcal{W} (e^{\frac{\lambda_o \pi h^2}{2}}\frac{\lambda_o \pi h^2}{2 \beta})}{2 h^4 \pi \lambda_s},\quad
\lambda_s=\frac{-\lambda_o \pi h^2 + 2 \mathcal{W} (e^{\frac{\lambda_o \pi h^2}{2}}\frac{\lambda_o \pi h^2}{2 \beta})}{2 h^4 \pi Z_1}, \quad \mbox{and} \nonumber
\end{equation*}
\begin{equation}
\lambda_o=\frac{2 \beta e^{h^4 \pi Z_1 \lambda_s} h^2  Z_1 \lambda_s}{1-\beta e^{h^4 \pi Z_1 \lambda_s} }, \nonumber
\end{equation}
where $\mathcal{W}(\cdot)$ is the Lambert-W function.
\end{corollary}
Note that $Z_1$ is a function of various network parameters of both RF and VLC networks such as transmit power of SBSs $P_s$ and OBSs $P_o$, respectively. Also, the parameters like $m$ which sets the LED illumination angle can be designed using $Z_1$ accordingly.

Since VLC networks are expected to be denser than RF networks over indoor environments, it is interesting to determine an approximation of the association probability when $\lambda_s \rightarrow 0$. Given this condition, we  show in the following that the network parameters can be optimized to prioritize VLC networks over RF network by maximizing the association probability.
\begin{corollary}[Asymptotic Association Probability, $\lambda_s\rightarrow 0$] The asymptotic approximation of the association probability when $\lambda_s\rightarrow 0$ can be given using \eqref{approx2} as follows:
\begin{equation}
\mathbb{P}_o=\frac{\lambda_o}{\lambda_o+2 h^2 Z_1 \lambda_s}-\frac{\lambda_o e^{-\pi \lambda^2_o \mathcal{T}^2}}{\lambda_o+2 h^2 Z_1 \lambda_s}.
\end{equation}
The closed-form optimal solution (to maximize the association probability) for $Z_1$  is:
\begin{equation}
Z_1^*=\frac{\lambda_o \mathrm{cos}(\xi_{\mathrm{fov}}) (1+e^{0.5 \pi \lambda_o \mathcal{T}^2} \mathrm{cos}(\xi_{\mathrm{fov}}))}{2 h^2 \lambda_s (\mathrm{cos}(\xi_{\mathrm{fov}})+e^{0.5 \pi \lambda_o \mathcal{T}^2})}.
\end{equation}
Also, the traffic load can be distributed by considering $\mathbb{P}_o=\beta$ and selecting the system parameters as:
\begin{equation}
\lambda_s=\frac{\lambda_o -\beta \lambda_o-\lambda_o e^{-\mathcal{T}^2 \pi \lambda_o^2}}{2 h^2 Z_1 \beta} \qquad \text{and} \qquad Z_1=\frac{\lambda_o -\beta \lambda_o-\lambda_o e^{-\mathcal{T}^2 \pi \lambda_o^2}}{2 h^2 \lambda_s \beta}.
\end{equation}
\end{corollary}

\subsubsection{Distance Distributions}
In order to derive $\tilde{\mathcal{C}}_{o}$ and $\tilde{\mathcal{C}}_{s}$, we first derive the distribution of the distance of the selected OBS or SBS from the typical user, denoted by $X_o$ and $X_s$, respectively. Note that this distance is different than the nearest distance of a BS in a specific tier as it depends on the association criterion. The relevant distance distributions can then be given as in the following.

\begin{lemma}[Distributions of $X_o$ and $X_s$]
Conditioned on the fact that user selects OBS, the PDF $f_{X_o}(x)$ of the distance $X_o$ between a typical user and its serving OBS is given by
\begin{equation}
f_{X_o}(x_o)=\frac{2 \pi \lambda_o x_o}{\mathbb{P}_o}   e^{-\lambda_s \pi Z_1 \left(x_o^2+h^2 \right)^{\frac{2(m+3)}{\alpha}} -\lambda_o \pi x_o^2}.
\end{equation}
Similarly, if the typical user associates to an SBS,  the PDF $f_{X_s}(x)$ of the distance $X_s$ between a typical user and its serving SBS can be derived as in the following:
\begin{equation}
f_{X_s}(x_s)=
\frac{2 \pi \lambda_s x_s} {1-\mathbb{P}_o} 
\left(\mathbb{P}( r \leq \mathcal{T})e^{-\lambda_o \pi \left({\left({x_s Z_1^{-1/2}}\right)^{\frac{\alpha}{m+3}}-h^2} \right)-\lambda_s \pi x_s^2} + (1-\mathbb{P}( r \leq \mathcal{T}))e^{-\lambda_s \pi x_s^2}\right).
\end{equation}
\end{lemma}
\begin{proof}
Given the typical user associated to OBS, the PDF of $X_o$ can be derived as follows:
\begin{align*}
\mathbb{P}(X_o\geq x_o)&=\mathbb{P}(r\geq x_o, r \leq \mathcal{T}| {\mathrm{OBS \quad selected}})
=\frac{\mathbb{P}(r\geq x_o, r \leq \mathcal{T},{\mathrm{OBS \quad selected}})}{\mathbb{P}_o}
\\&=\frac{\mathbb{P}(r\geq x_o,  r \leq \mathcal{T}, P^{\mathrm{vlc}} (r) >  P^{\mathrm{rf}} (u))}{\mathbb{P}_o},
\\&=
\frac{1}{{\mathbb{P}_o}}{\int_{x_o}^{\mathcal{T}} \mathrm{exp}\left(-\lambda_s \pi \left(\frac{P_s K  (r^2+h^2)^{{m+3}}}{Z R^2_{\mathrm{pd}}P_o} \right)^{\frac{2}{\alpha}}\right) f_{r}(r) dr}.
\end{align*}
The CDF of $X_o$ can thus be given as $F_{X_o}(x)=1-\mathbb{P}(X_o\geq x_o)$. Subsequently, the PDF can be derived by differentiating the CDF as given in {\bf Lemma~7}.
Now, given the typical user associated to an SBS, there are two possibilities, (i) at least one OBS exists with in FOV, i.e., $r \leq \mathcal{T}$ OR (ii)  all OBSs exist outside FOV, i.e., $r > \mathcal{T}$. The distance distribution  of $X_s$ can be derived as follows:
\begin{align}\label{step1}
\mathbb{P}(X_s\geq x_s)&=\mathbb{P}(\{v\geq x_s,r \leq \mathcal{T} \bigcup v\geq x_s,r > \mathcal{T}\}|\mathrm{SBS\quad selected}),
\\&=\frac{\mathbb{P}(v\geq x_s, r \leq \mathcal{T}, P^{\mathrm{rf}} (v) >  P^{\mathrm{vlc}} (r))+\mathbb{P}(v\geq x_s, r > \mathcal{T}, P^{\mathrm{rf}} (v) > 0 )}{1-\mathbb{P}_o}.
\end{align}
Note that $P^{\mathrm{vlc}} (r) >0$ iff $r \leq \mathcal{T}$, therefore  $\mathbb{P}(P^{\mathrm{rf}} (v) > P^{\mathrm{vlc}} (r) )$ can be given as follows:
\begin{align}
\mathbb{P}(P^{\mathrm{rf}} (v) > P^{\mathrm{vlc}} (r) )=& 
\mathbb{P}\left(r > \sqrt{\left(v Z_1^{-1/2}  \right)^{\frac{\alpha}{m+3}}-h^2} =\mathcal{G}(v)\right)
=
{e}^{-\lambda_o \pi (\mathcal{G}(v))^2},  \mathcal{G}(v) \leq \mathcal{T}
\end{align}
The condition $\mathcal{G}(v) \leq \mathcal{T}$ occurs only when $v \leq \sqrt{Z_1} (\mathcal{T}^2 +h^2)^{\frac{m+3}{\alpha}}=\mathcal{G}_1$. Consequently,
\begin{equation}\label{step2}
\mathbb{P}(v\geq x_s, r \leq \mathcal{T}, P^{\mathrm{rf}} (v) >  P^{\mathrm{vlc}} (r)) = 
\int_{x_s}^{\mathrm{min}(R_m,\mathcal{G}_1)}
{e}^{-\lambda_o \pi (\mathcal{G}(v))^2} 
f_{v}(v) dv,\quad   \mathcal{G}(v) \leq \mathcal{T}.
\end{equation}
When $r > \mathcal{T}$, all OBSs exist outside FOV. That is $\mathbb{P}(P^{\mathrm{rf}} (v) > 0)=1$.
$\mathbb{P}(v\geq x_s, r > \mathcal{T}, P^{\mathrm{rf}} (v) > 0)$, can then be given as follows:
\begin{equation} \label{step3}
\mathbb{P}(v\geq x_s, r > \mathcal{T}, P^{\mathrm{rf}} (v) > 0)=
\int_{x_s}^{R_m}
f_{v}(v) dv.
\end{equation}
Substituting \eqref{step2} and \eqref{step3} in \eqref{step1} and then differentiating over $v$, $f_{X_s}(x)$ can be given as in {\bf Lemma~7}.
\end{proof}

\subsubsection{Coverage Probability} 
The coverage probabilities $\tilde{\mathcal{C}}_o$ and $\tilde{\mathcal{C}}_s$ can thus be calculated by averaging over  the distance distributions $f_{X_o}(x)$ in \eqref{phi} and $f_{X_s}(x)$ in \eqref{cov1} instead of $f_r(r)$ and $f_v(v)$, respectively.

\subsubsection{Average Spectral Efficiency} 
The average spectral efficiency of a typical user with opportunistic RF/VLC scheme can be given as follows:
\begin{equation}
\mathcal{R}_H=\mathbb{P}_o \mathcal{R}_o+(1-\mathbb{P}_o)\mathcal{R}_s,
\end{equation}
where $\mathcal{R}_s$ and $\mathcal{R}_o$ can be given using \eqref{Rs} and \eqref{Ro} with the modification of taking expectation over $X_s$ and $X_o$ instead of $V$ and $R$, respectively.
 
\begin{table}[t]
\centering
\caption{Main parameters}
     \begin{tabular}{ | l  | p{2cm} | }
     \hline
		\hline
\textbf{Variable} & \textbf{Value} 
\\
\hline
Transmit optical power per OBS $(P_{\mathrm{opt}})$  & 10~W
\\
\hline
Modulation bandwidth for OBS $(B_o)$ &  40~MHz
\\
\hline
Physical area of PD $(A_{\mathrm{pd}})$ & 0.0001~$m^2$
\\
\hline
Target rate $(R_{\mathrm{th}})$ & 3~bps
\\
\hline
Intensity of SBSs $(\lambda_s)$ & 5
\\
\hline
Gain of optical filter $(T(\xi))$ & 1
\\
\hline
RF path-loss exponent $(\alpha)$ & 3.68
\\
\hline
Refractive index $(n)$ & 1.5
\\
\hline
Optical to electric conversion efficiency $(R_{\mathrm{pd}})$ & 0.53~A/W
\\
\hline
Noise power spectral density $(N_o)$ & $10^{-21}$~A$^2$/Hz
\\
\hline
RF bandwidth $(B_s)$ & 10~MHz
\\
\hline
RF transmit power $(P_s)$ & 2~W
\\
\hline
Optical-to-electrical conversion ratio $(f)$ & 3
\\
\hline
Height of the room $(h)$ & 2~m
\\
\hline
\hline
\end{tabular}
\end{table}

\section{Numerical Results and Discussions}
In this section, we analyze the impact of system parameters such as intensity of OBSs and the FOV of the PD receiver on  the PMF of the number of interferers, distribution of the aggregate interference, and coverage probability of a typical user in various network settings. We comparatively analyze various network configurations such as RF-only, VLC-only, opportunistic RF/VLC, and hybrid RF/VLC networks. Simulation parameters are listed in Table~I.

\subsection{VLC-Only Network Configuration}

\subsubsection{PMF of the Number of Interferers}
\figref{1} demonstrates the impact of the FOV of the PD receiver on the PMF of the number of potential interferers considering $\lambda_o|A|=30$. The analytical results closely follow the histogram generated by the Monte-Carlo simulations. It can be seen that only a small increase in  the FOV  (i.e., $\xi_{\mathrm{fov}}=70^\circ$ to $\xi_{\mathrm{fov}}=80^\circ$) significantly shifts the mean of the distribution which represents the average number of interferers inside the FOV of the PD receiver. Moreover, with an increase in intensity of OBSs, the number of interferers inside the FOV of the PD receiver increases significantly as shown in \figref{2} where cumulative distribution is plotted for various intensities of OBSs. The Monte-Carlo simulations verify the theoretical results in \eqref{no}.

\begin{figure}[h]
\centering
\begin{minipage}{0.48\textwidth}
\centering
\includegraphics[scale=0.5]{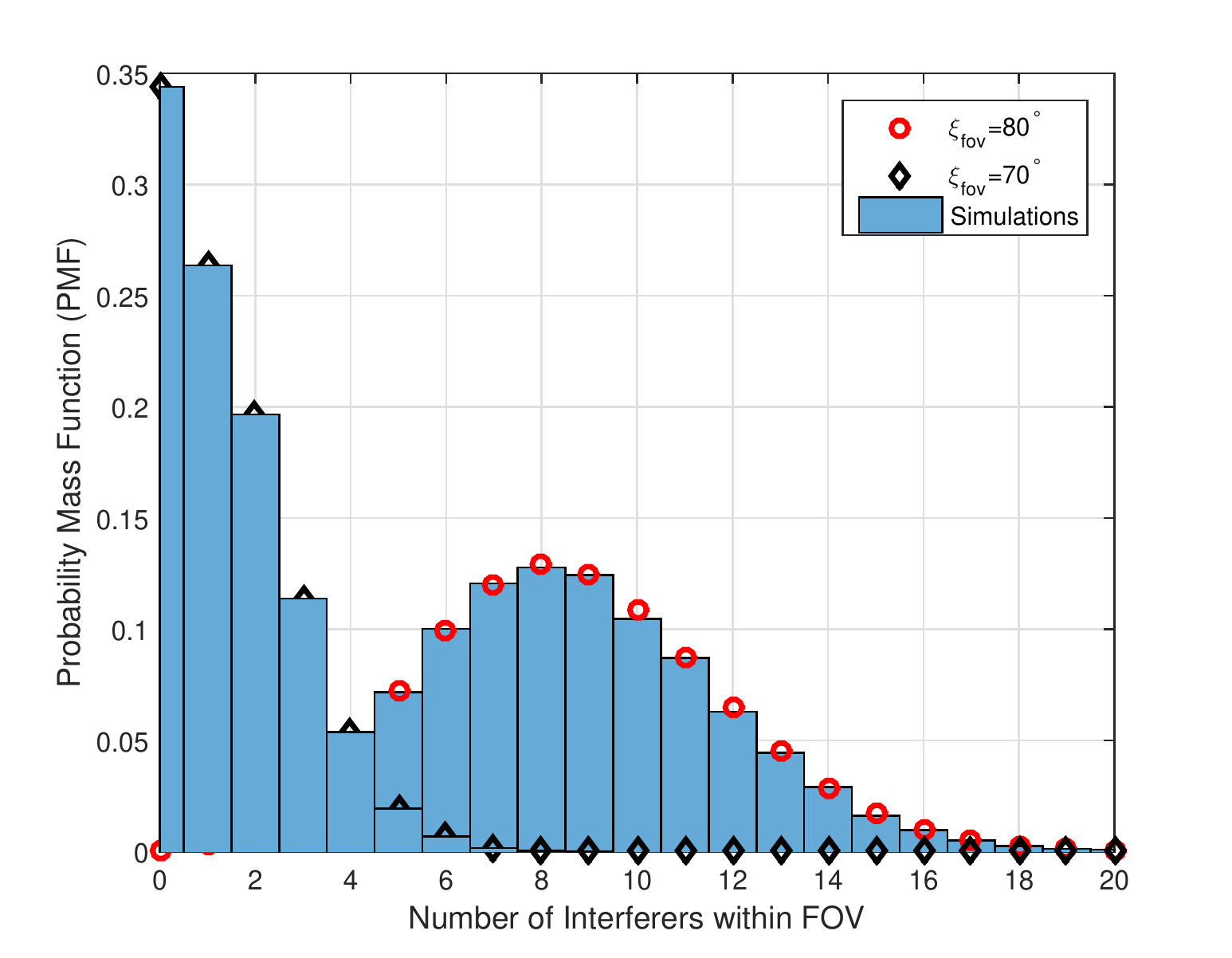}
\caption{PDF of the number of interferers existing within the FOV of the PD receiver mounted at the typical user as a function of the FOV of the PD receiver (for $\lambda_o|A|=30$).}
\label{1}
\end{minipage}
\hfill
\begin{minipage}{0.48\textwidth}
\centering
\includegraphics[scale=0.5]{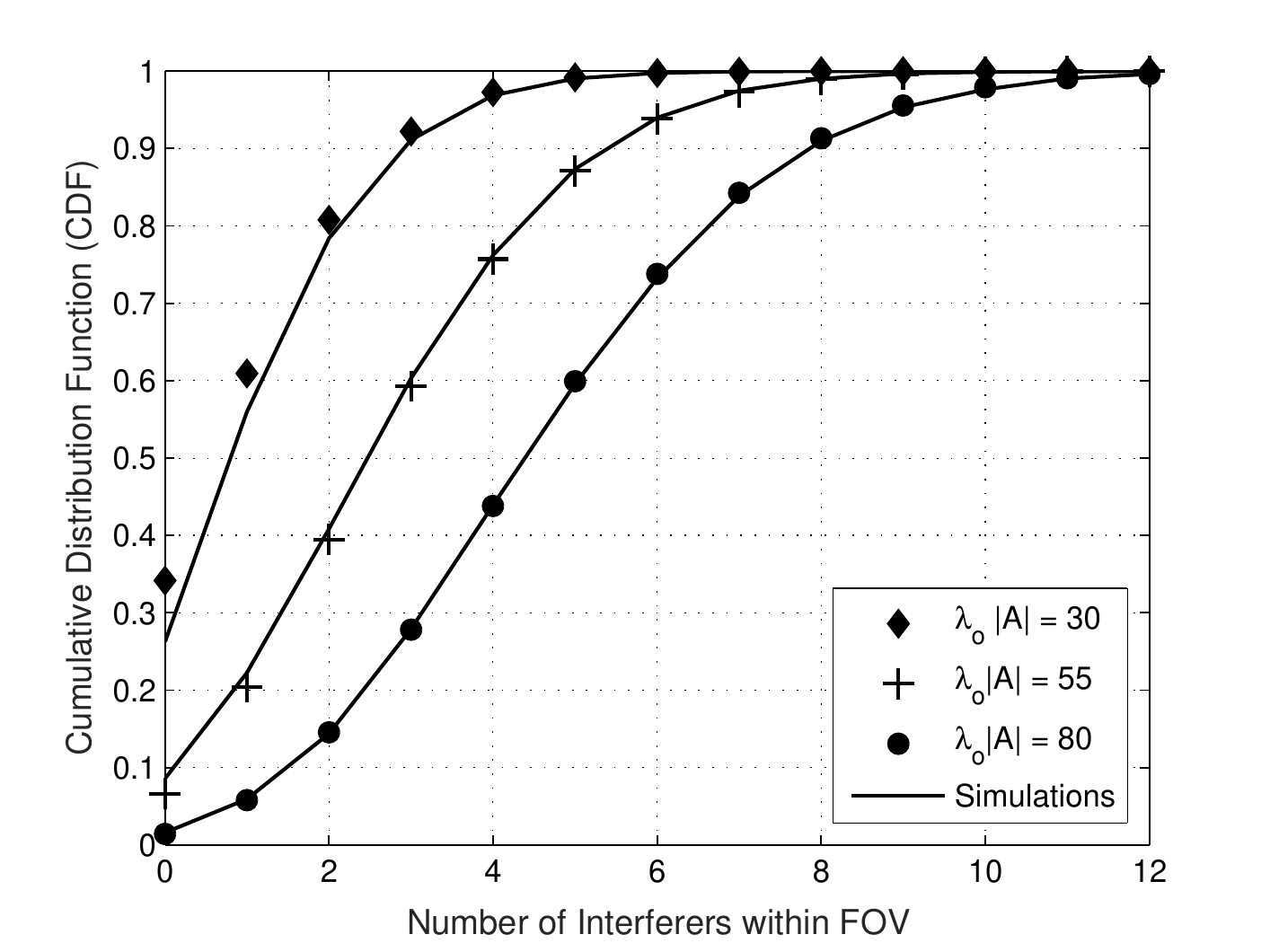}
\caption{CDF of the number of interferers  within the FOV of the PD receiver mounted at the typical user as a function of the intensity of OBSs (for $\xi_{fov}=70^\circ$).}
\label{2}
\end{minipage}

\end{figure}

\subsubsection{Laplace Transform of Aggregate Interference}

\figref{BL} demonstrates the impact of the FOV of the PD receiver on the Laplace Transform of the cumulative interference considering $\lambda_o|A|=80$. This plot shows the unconditional Laplace Transform, i.e., $\mathbb{E}_r\left[\mathcal{L}_{I|r}(t)\right] \mathbb{P}(r \leq \mathcal{T})+\mathbb{P}(r \geq \mathcal{T})$ where $f_r(r)$ is given by the Rayleigh distribution. The first term is dominant in scenarios when there is a high probability of an interferer to exist within the FOV of the receiver (e.g., large FOV or intensity of OBSs). The second term is dominant for scenarios when there is no  interferer, i.e., zero interference which makes the Laplace Transform unity. By definition, the value of the Laplace Transform $\mathcal{L}_{I}(s)=\mathbb{E}(e^{-s I})$ of a random variable reduces with increasing $s$. That is, as $s \rightarrow \infty $, $\mathcal{L}_{I}(s)\rightarrow 0$. 
Subsequently,  $I$ is also inversely proportional to $\mathcal{L}_{I}(s)$. Therefore, it can be noted that the cumulative interference increases significantly with increasing FOV. For low values of FOV, it can be seen that the aggregate interference becomes nearly constant after a certain value of $s$. The reason is that, in such cases, the probability of no interferer to exist within the FOV of the PD receiver $\mathbb{P}(r\leq \mathcal{T})$ is high which is independent of $s$ and in such a case $\mathcal{L}_{I}(s)$ is unity. 
\begin{figure}[ht]
\centering
\includegraphics[scale=0.6]{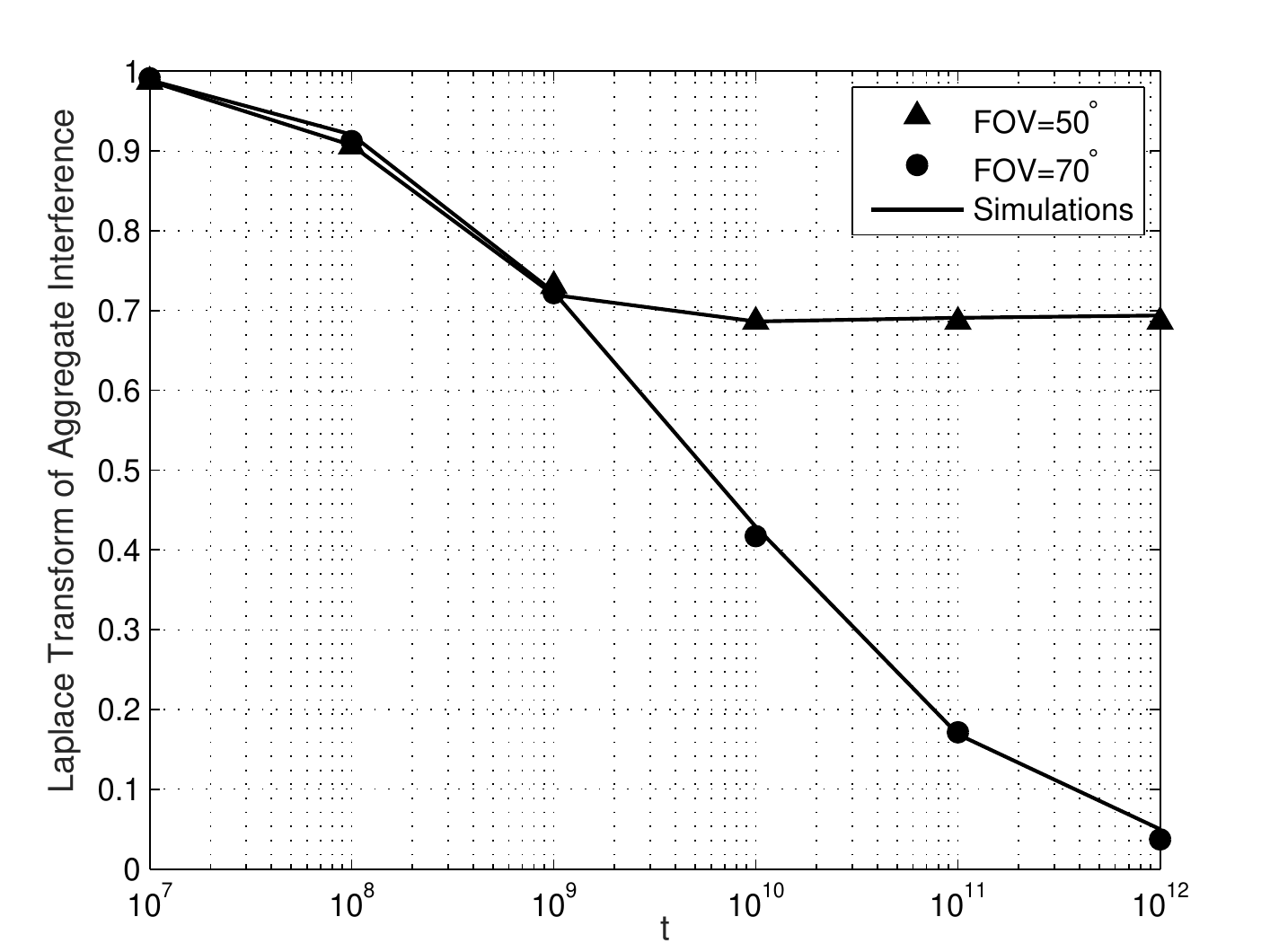}
\caption{Laplace Transform of the aggregate interference incurred at a  typical user as a function of the FOV of  PD receiver (for $\lambda_o |A|$ = 80).}
\label{BL}
\end{figure}

\subsubsection{Coverage Probability of a Typical User}
\figref{BL1} depicts the coverage probability of a typical user in VLC-only network as a function of the FOV of the PD receiver. It can be seen that as FOV increases, the coverage probability continues to increase since there is a high probability of detecting and associating to an OBS inside FOV. Although the interference also increases with the increase in FOV, the performance gain due to higher association and transmission probability dominates the performance loss due to higher interference. Moreover, it can also be observed that the coverage probability significantly varies as a function of the FOV and the height of OBSs. That is, the higher the  deployment height of OBSs, the larger is the coverage probability. The reason is evident from the condition derived  in {\bf Lemma~2} which shows that the probability of getting an OBS within the FOV of the PD receiver is a function of both height and FOV. That is, either increasing FOV  or height will allow more OBSs within the FOV of the receiver which ultimately enhances the coverage. 
Nonetheless, for a given FOV, increasing height beyond a certain threshold may not be beneficial as is depicted in the  next figure. 
\begin{figure}[h]
\centering
\begin{minipage}{0.48\textwidth}
\centering
\includegraphics[scale=0.6]{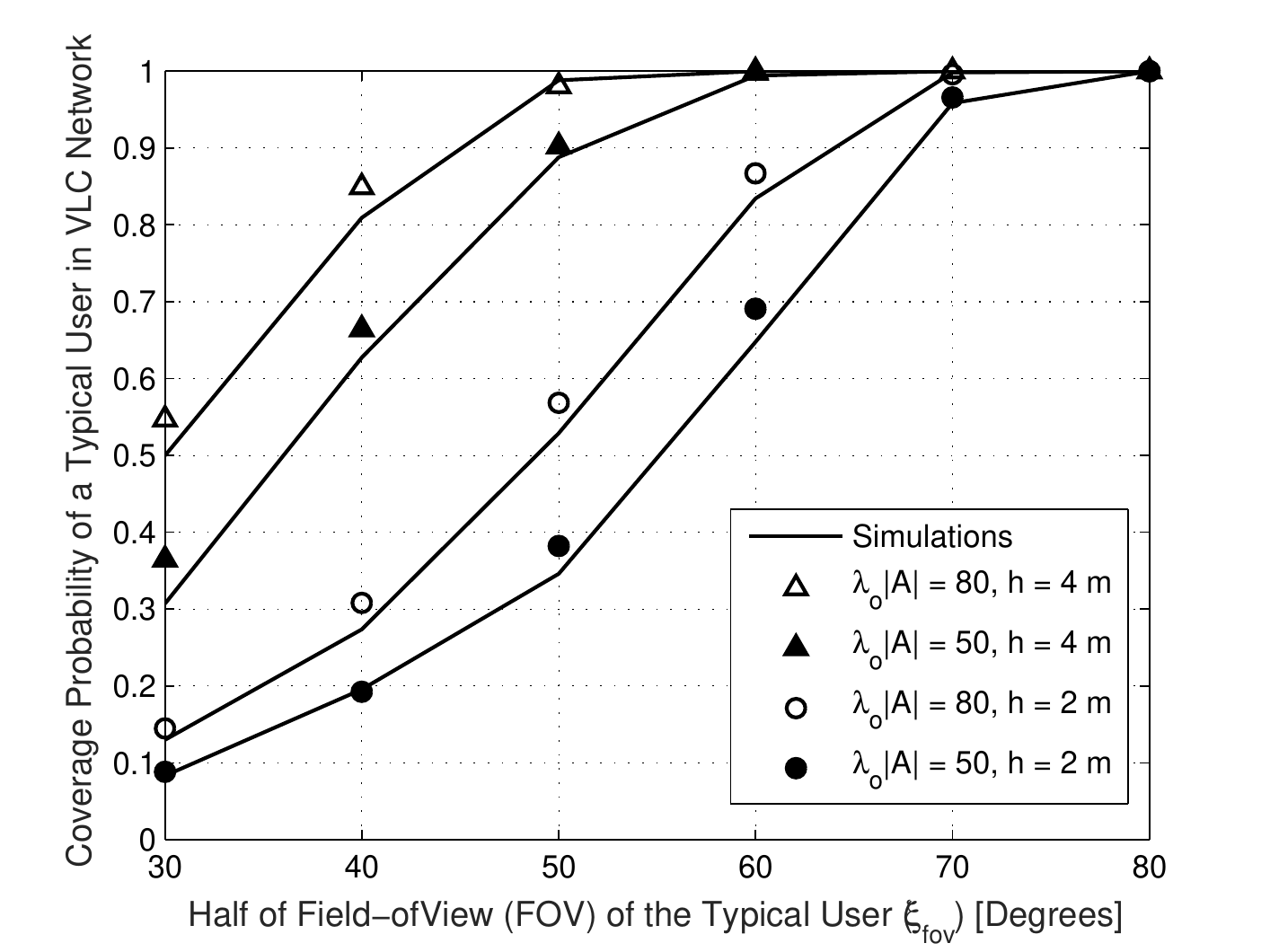}
\caption{Coverage probability of a typical user as a function of the FOV of  PD receiver, height and intensity of OBSs.}
\label{BL1}
\end{minipage}
\hfill
\begin{minipage}{0.48\textwidth}
\centering
\includegraphics[scale=0.6]{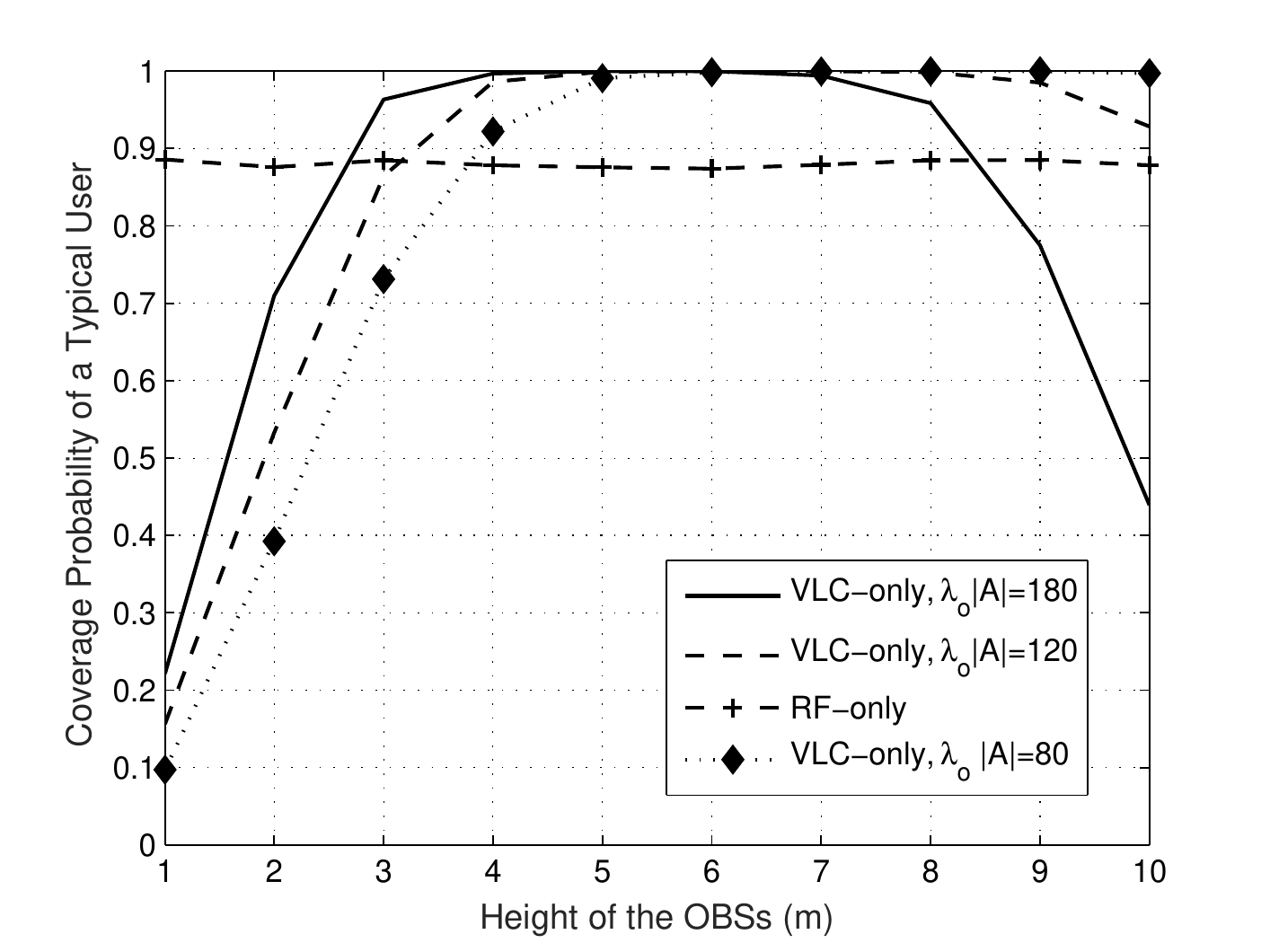}
\caption{Coverage probability of a typical user as a function of the height and intensity of OBSs (for $\xi_{\mathrm{fov}}=45^\circ$, $\alpha=3.68$).}
\label{height1}
\end{minipage}
\end{figure}

\figref{height1} depicts the coverage probability of a typical user as a function of the height and intensity of OBSs. It can be seen that as the height increases, the coverage probability first increases due to higher chances of getting an OBS within FOV of PD receiver and thus higher probability of association/transmission. However, beyond a certain height, the path-loss degradation due to increasing distance becomes more dominant leading to reduction in coverage. Also, the number of OBSs within FOV also becomes too high to make  interference more dominant that results in coverage reduction. 
Also, we note that the range of the values of room height that maximizes coverage is a function of the intensity of OBSs. For example, the higher intensity of OBSs reduces the range of heights at which coverage can be maximized whereas for sparse OBSs the optimal coverage can be achieved for a wide range of heights. The reason is that, increasing the intensity of OBSs allows interference  to begin dominate rapidly.  Finally, it can also be observed that the VLC network outperforms the RF network only for a certain range of deployment heights or intensities. Therefore, it is crucial to select the correct intensity for a given deployment height of OBSs in order to harness the gains of VLC over RF networks.

\subsection{Opportunistic RF/VLC Network Configuration}

\subsubsection{Probability of Association to an OBS}
\begin{figure}[h]
\centering
\begin{minipage}{0.48\textwidth}
\centering
\includegraphics[scale=0.6]{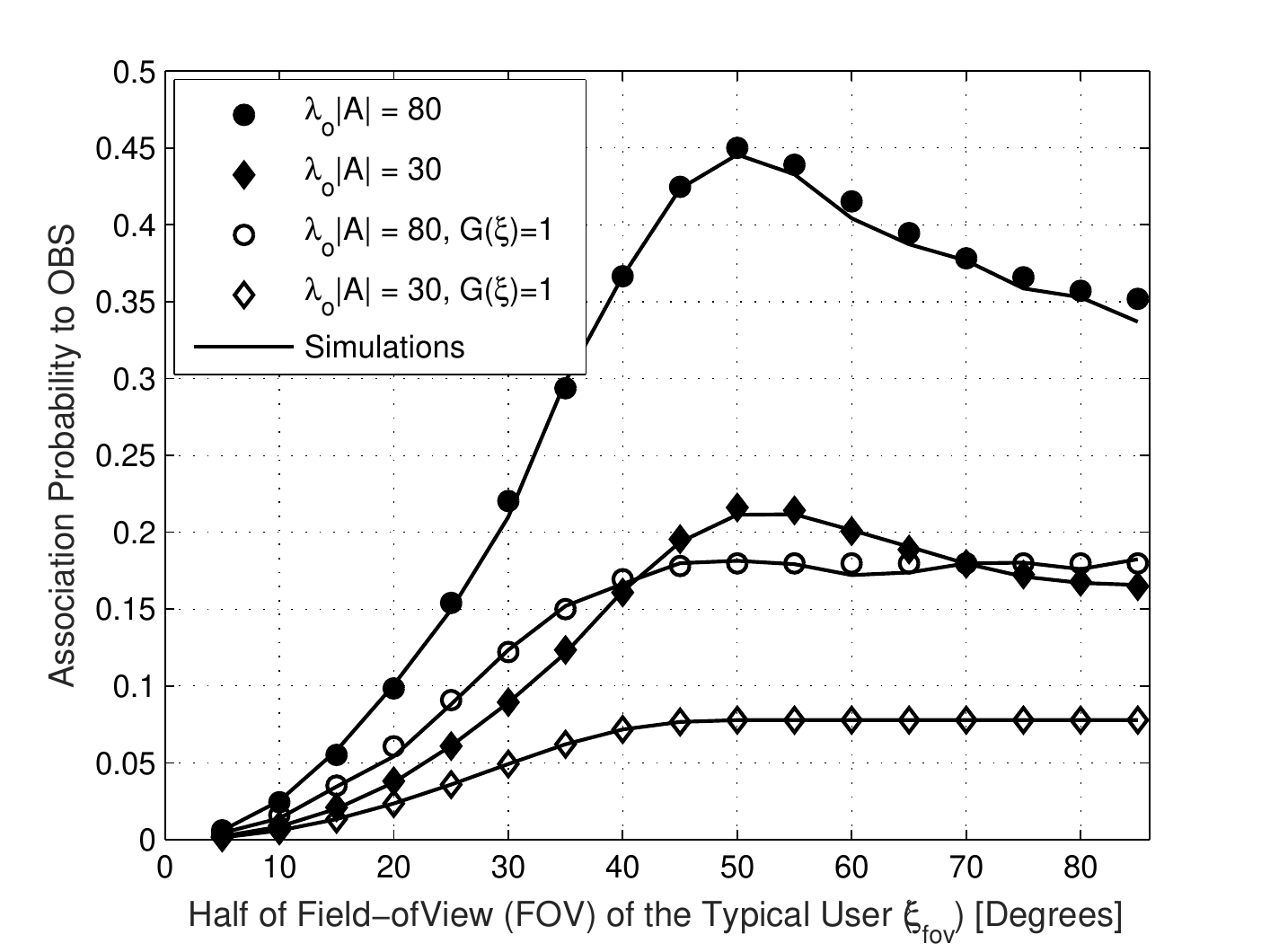}
\caption{Association probability of a typical user with OBS as a function of field-of-view (FOV) of the optical receiver (for $\alpha=3.68$).
}
\label{atto1}
\end{minipage}
\hfill
\begin{minipage}{0.48\textwidth}
\centering
\includegraphics[scale=0.5]{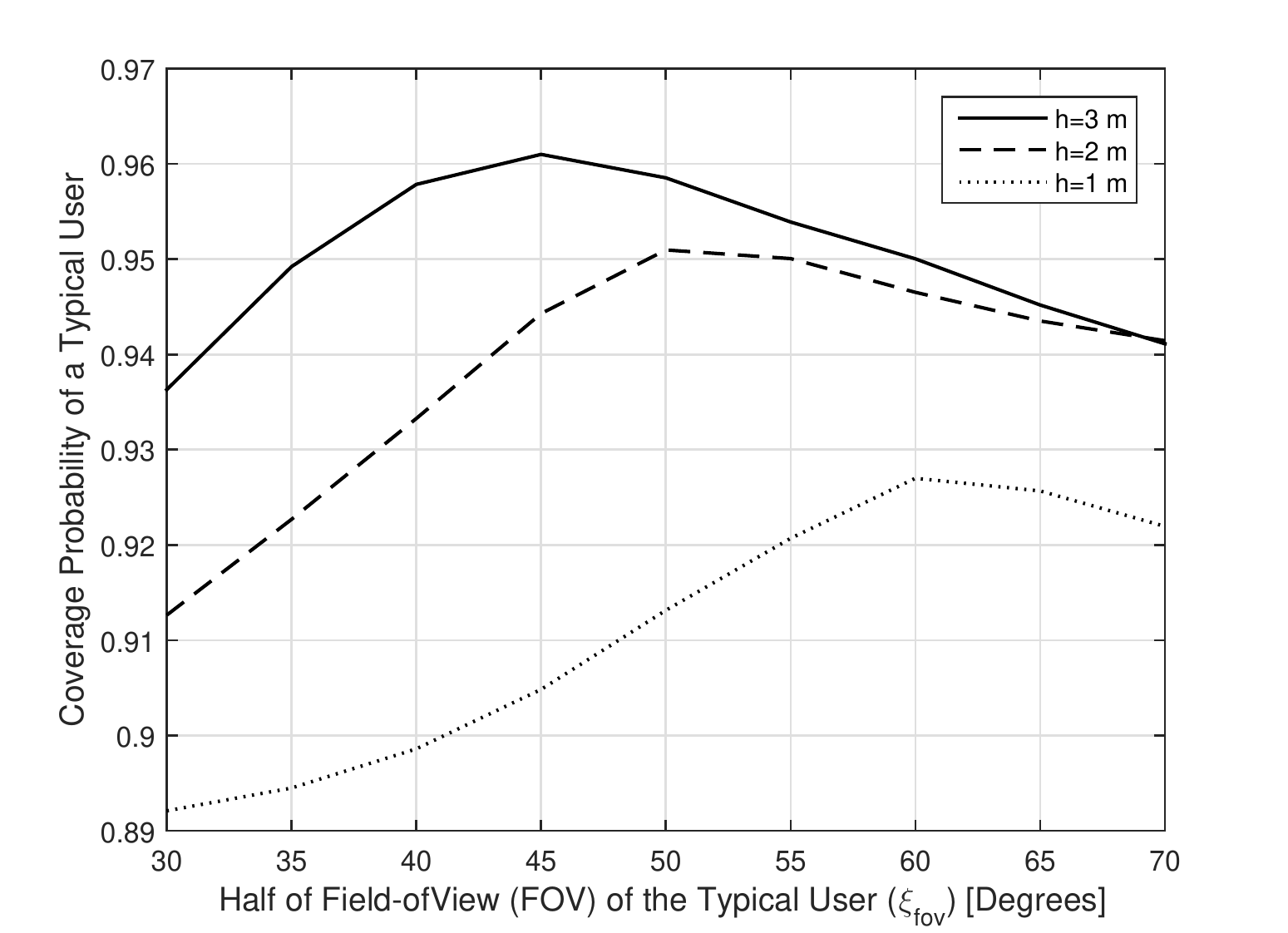}
\caption{Coverage probability of a typical user with OBS as a function of field-of-view (FOV) of the optical receiver and height of OBSs (for $\alpha=3.68$).
}
\label{atto2}
\end{minipage}

\end{figure}
\figref{atto1} depicts the association probability of a typical user with OBSs as a function of the FOV of its PD receiver. It can be observed that, as the FOV increases, the association to OBSs increases. However, this increase may not be significant if the intensity of OBSs is smaller as can be seen for the case $\lambda_o|A|=30$. Moreover, it can also be noted that beyond a certain FOV, the association probability to OBSs starts to decrease. This effect is observed mainly due to the  gain of the optical concentrator $G(\xi)$ which varies with the FOV of the PD receiver. However, if $G(\xi)=1$ as is considered in state-of-the-art performance analysis studies \cite{chen34downlink}, the decaying trend cannot be observed. The probability of association to the OBSs gives direct insight into the traffic load at each network since the traffic load at RF and VLC networks, respectively, can be written as $\lambda_u(1-\mathbb{P}_{o})$ and $\lambda_u \mathbb{P}_{o}$. 

\subsubsection{Coverage Probability}
\figref{atto2} depicts the coverage probability of a typical user in co-existing RF/VLC  networks with opportunistic network selection as a function of the FOV of the typical user. The opportunistic selection, i.e., based on maximum received signal power, is more suitable for scenarios where the interference effects are not dominant. For example, as the the value of FOV increases from low to moderate,
opportunistic selection increases the coverage probability. 
Nonetheless, beyond a certain FOV, the interference in VLC network becomes dominant and the coverage probability starts to reduce with the opportunistic selection. The reason is that the opportunistic selection is not an interference-aware scheme.

Interestingly this coverage reduction (with increasing FOV) cannot be observed in VLC-only network. The reason is that, at higher values of FOV,  the benefits from reduction of transmission outage events outweigh the coverage degradation due to increased interference. Note that transmission outages in co-existing RF/VLC networks are significantly low  due to the presence of RF network. Therefore, the coverage degradation can be observed for higher values of FOV. Moreover, as also depicted, in VLC-only network,  the optimal FOV tends to decrease as the height of OBSs increases since a larger height invites more interference and higher path-loss.

\begin{figure}
\centering
\includegraphics[scale=0.6]{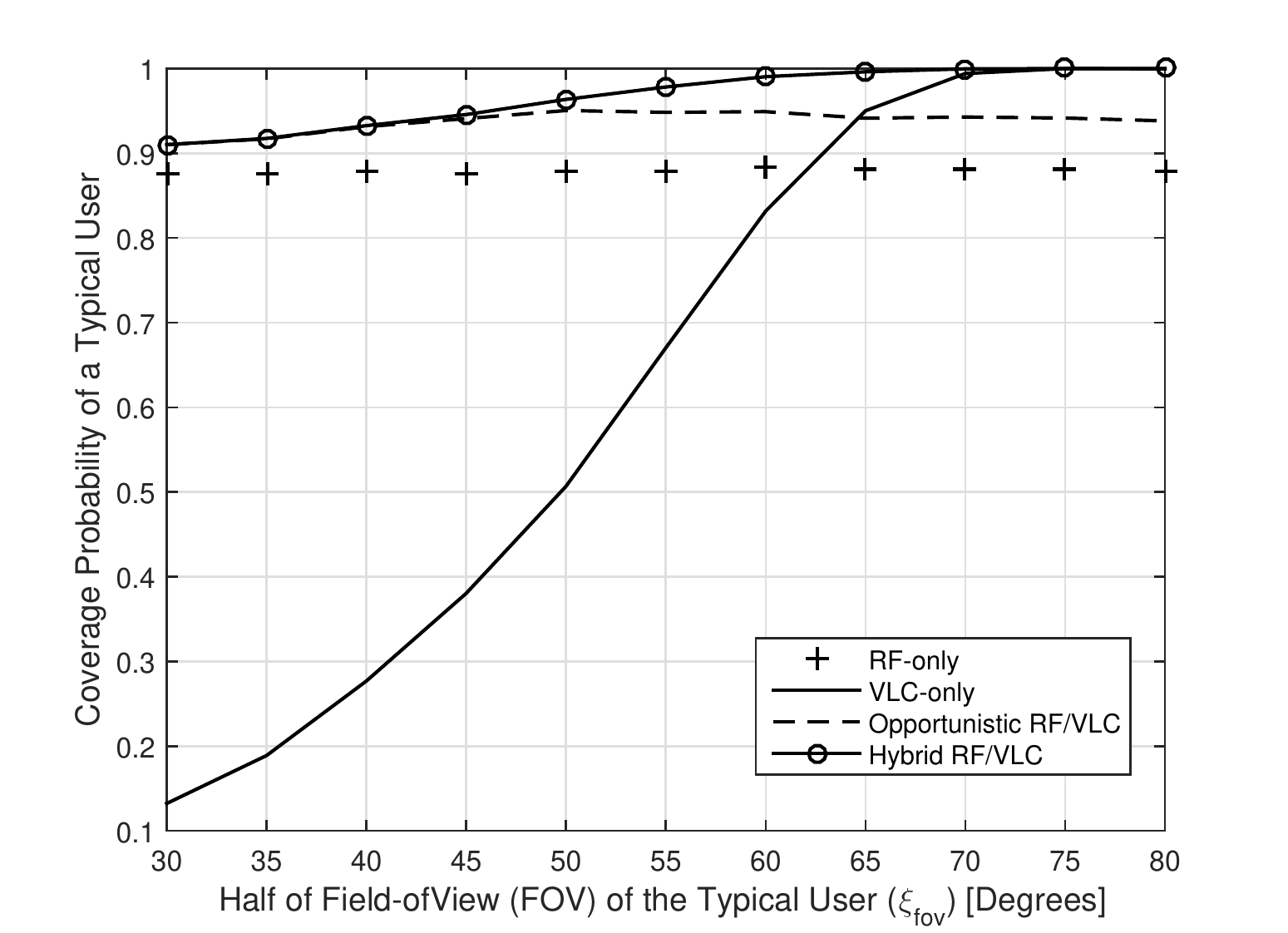}
\caption{Coverage probability of a typical user with OBS as a function of field-of-view (FOV) of the optical receiver and height of OBSs (for $\alpha=3.68$).
}
\label{atto3}
\end{figure}
\subsection{Comparative Analysis: All Configurations}
\figref{atto3} compares the performance of various network configurations, i.e., RF-only, VLC-only, opportunistic RF/VLC, and hybrid RF/VLC.
We assume that each user has multi-homing capability that allows four modes of data transmission. As expected, the hybrid scheme outperforms all other schemes at the expense of extra resources. On the other hand, opportunistic RF/VLC tends to outperform the isolated VLC networks especially in scenarios with low FOV. The reason is that, in low FOV scenarios, there may not be any OBS  for desired transmission. As such, association to RF network becomes beneficial. However, as FOV increases, the opportunistic scheme may suffer due to the use of RF resources based on maximum received signal power and ignoring the interference at RF channel. In such a case, an isolated VLC network tends to outperform opportunistic RF/VLC as well as an RF-only configuration.

\section{Possible Extensions}
In this section, we will briefly provide guidelines to extend the framework for more comprehensive network models, adapting to  mm-wave communication channel, and/or including the impact of blockages.
\subsection{Extension to Binomial Point Process Model for OBSs}
Since LEDs/OBSs can likely be clustered in a small indoor area, Binomial point process (BPP) may also be an interesting analytical model for VLC network.
Such a network model  can be considered as a special case of this framework in which a fixed number of OBSs (say $N$) will be uniformly distributed in the circular region of radius $R_m$. As such, the number of interferers can be determined as $N_I=N\frac{(R_m^2 -r^2)}{R_m^2}$ where $r$ represents the distance of serving OBS. Since all interferers are i.i.d, the Laplace Transform (conditioned on $r$) of $\mathcal{I}_a(s)$ can be given as $(\mathcal{L}_I(s))^{N_I}$ where $\mathcal{L}_I(s)$ can be given as in {\bf Lemma~5}. Moreover, the distribution of distance of serving BS $r$ can now be given by \cite[Lemma~1]{noma}. Note that the BPP model is a special case of Matern cluster process (MCP) with a single cluster. The coverage probability of a typical user can then be given using Gil-Pelaez inversion as in {\bf Lemma~4}.

Furthermore, to consider  several rooms with clustered LEDs, we can use a full-fledged MCP or a modified Thomas cluster process~\cite{stoch}. Note that, each cluster represents a separate room thus  different clusters are unlikely to interfere with each other unless they overlap (with no blockages and/or have no boundary walls). In  contrast  to  a  homogenous  PPP,  the  location of  the user of interest  is  crucial  in cluster processes to characterize its corresponding coverage probability.

\subsection{Incorporating NLoS Reflection Components}
In addition to LoS reception, users may also receive VLC data through the reflected paths (e.g., owing to the wall, floor, human obstacles, etc.). The  channel gain via one reflector can be given as follows~\cite{reflection}:
\begin{equation}
G_{r}= \frac{ (q+1) A_{\mathrm{pd}} T(\xi) G(\xi) \rho}{2 \pi^2 D_1^2 D_2^2} dA_{\mathrm{wall}} \mathcal{U}, \quad \xi \leq \xi_{\mathrm{fov}},
\end{equation}
where $q$ denotes the LoS blocking probability, $D_1$ is the distance between an LED light and a reflecting surface, $D_2$ denotes the distance between a reflective point and user, $\rho$ is the reflectance factor, $\mathcal{U}=\mathrm{cos}(\omega_1) \mathrm{cos}(\omega_2)$, $\omega_1$ and $\omega_2$ denote the angle of irradiance to a reflective point and user, respectively, and $dA_{\mathrm{wall}}$ is a small reflective area. 
Given that $\mathbb{P}(r \leq \mathcal{T})$ and defining $X_1=R_{\mathrm{pd}}^2 P_{o} G^2_r$, we can calculate the coverage probability of a typical user as given in \textbf{Lemma~2} and defining $\Phi_\Omega (\omega)$ as $
\phi_{\Omega}(\omega)=
\mathbb{E}_{r}[\phi_{\Omega|r}(\omega)]
=\mathbb{E}_{r}\left[e^{-j \omega (X+X_1)}  \mathcal{L}_{\mathcal{I}_a|r}({-j \omega \tilde\gamma^{\mathrm{vlc}}})\right],
$
where, conditioned on $r$, $X$ and $X_1$ are independent. As the number of reflected components keeps increasing, the term $e^{-j \omega (X+X_1+\cdots)}$ can be updated accordingly.

\subsection{Incorporating Shadowing in VLC networks}
In wireless networks, the  log-normal  distribution  has  been  considered as the  most suitable statistical  model for  shadowing  effects.  We model shadowing with the log-normal distribution
where $\mu$ and $\sigma$ are the mean and standard deviation
of the shadowing channel power, respectively. 
The displacement
theorem can be used to deal with the shadow fading 
as a random and independent transformation of a given homogeneous
PPP of density $\lambda$. In this case, the resulting
point process is also a PPP with equivalent density $\lambda\mathbb{E}[S^{\frac{2}{\beta}}]$.
Applying this theorem to our case, we can handle the effect of
any distribution for the shadow fading as long as the fractional
moment $\mathbb{E}[S^{\frac{2}{\beta}}]$ is finite. 

\subsection{Extension to mm-wave RF networks}
This framework can be modified for mm-wave RF networks (assuming only LoS transmissions) since the scattered components of the transmitted signal are extremely weak, and in most scenarios can be neglected~\cite{shao2015design}. The received signal power can be modeled as $P_i^{\mathrm{rf}}= K d^{-\alpha} \chi$, where $\alpha=1.6$, $K$~[dB] = 68~dB, and $X$ can be considered as Gamma-distributed shadowing\footnote{In wireless networks, log-normal  distribution  has  been  considered as the  most suitable statistical  model for  shadowing  effects.  However,  in  spite  of  its  usefulness,   when  log-normal  is  involved  in  combination  with  other  elementary  and/or
special  function,  its  algebraic  representation  becomes  intractable.    Motivated  by
this,  \cite{abdi}  proposes
Gamma distribution as an accurate substitute for log-normal distribution.}.
This shadowing can occur for LoS propagations if one or more
Fresnel zones are blocked by large objects or humans in indoor
environments while the geometric LoS path is not blocked.
As such, by neglecting the inter-cell interference which is nearly negligible for mm-wave networks, i.e., $\mathcal{L}_I(s)=1$ and customizing the path-loss model, we can extend our results to a simplified mm-wave networks.

\section{Conclusion}
This paper provides a unified framework for the exact coverage and rate analysis of coexisting RF/VLC networks under different network configurations such as RF-only, VLC-only, opportunistic RF/VLC, and hybrid RF/VLC networks. Using sophisticated  approximations for the CDF of a Gamma random variable and complementary error function ($\mathrm{erfc}$), approximate coverage expressions have also been derived. In order to balance the traffic load distribution as per the network designer requirements, we have derived the closed-form expressions of different network parameters such as $\lambda_o$, $Z_1$ that depend on $P_s$, $\Phi_{1/2}$ and FOV of the PD detector. The parameters  for specific scenarios of practical interest such as when RF BSs are sparse $\lambda_s \rightarrow 0$, have also been derived.
Numerical results show that the performance gains of VLC network can be guaranteed over the RF network only for a certain range of deployment heights or intensities. 
Therefore, it is crucial to select the correct intensity for a given deployment height of OBSs. 
Optimal FOV tends to decrease as the height of OBSs increases since a larger height invites more interference as well as higher path-loss.
The opportunistic selection, i.e., based on maximum received signal power, is more suitable for scenarios where the interference effects are not dominant, e.g., low FOV or low intensity of OBSs.
For higher interference scenarios, the opportunistic scheme deteriorates system performance due to wrongfully connecting to RF networks with higher interference instead of VLC networks. The presented framework can be extended to consider more sophisticated association and traffic load balancing schemes. Moreover, precise approximations for the interference statistics would also be of interest in simplifying the coverage probability expressions and optimize network parameters. 

\bibliography{Ref1}

\begin{thebibliography}{10}
\providecommand{\url}[1]{#1}
\csname url@samestyle\endcsname
\providecommand{\newblock}{\relax}
\providecommand{\bibinfo}[2]{#2}
\providecommand{\BIBentrySTDinterwordspacing}{\spaceskip=0pt\relax}
\providecommand{\BIBentryALTinterwordstretchfactor}{4}
\providecommand{\BIBentryALTinterwordspacing}{\spaceskip=\fontdimen2\font plus
\BIBentryALTinterwordstretchfactor\fontdimen3\font minus
  \fontdimen4\font\relax}
\providecommand{\BIBforeignlanguage}[2]{{%
\expandafter\ifx\csname l@#1\endcsname\relax
\typeout{** WARNING: IEEEtran.bst: No hyphenation pattern has been}%
\typeout{** loaded for the language `#1'. Using the pattern for}%
\typeout{** the default language instead.}%
\else
\language=\csname l@#1\endcsname
\fi
#2}}
\providecommand{\BIBdecl}{\relax}
\BIBdecl

\bibitem{1}
H.~Elgala, R.~Mesleh, and H.~Haas, ``Indoor optical wireless communication:
  Potential and state-of-the-art,'' \emph{IEEE Communications Magazine},
  vol.~49, no.~9, 2011.

\bibitem{2}
H.~Haas, L.~Yin, Y.~Wang, and C.~Chen, ``What is {LiFi}?'' \emph{IEEE Journal
  of Lightwave Technology}, vol.~34, no.~6, pp. 1533--1544, 2016.

\bibitem{3}
K.~Wang, A.~Nirmalathas, C.~Lim, and E.~Skafidas, ``Experimental demonstration
  of a novel indoor optical wireless localization system for high-speed
  personal area networks,'' \emph{Optics letters, Optical Society of America},
  vol.~40, no.~7, pp. 1246--1249, 2015.

\bibitem{feng2016applying}
L.~Feng, R.~Q. Hu, J.~Wang, P.~Xu, and Y.~Qian, ``Applying {VLC} in {5G}
  networks: Architectures and key technologies,'' \emph{IEEE Network}, vol.~30,
  no.~6, pp. 77--83, 2016.

\bibitem{khan2016visible}
L.~U. Khan, ``Visible light communication: Applications, architecture,
  standardization and research challenges,'' \emph{Digital Communications and
  Networks}, 2016.

\bibitem{gbps1}
A.~H. Azhar, T.~Tran, and D.~O'Brien, ``A gigabit/s indoor wireless
  transmission using {MIMO-OFDM} visible-light communications,'' \emph{IEEE
  Photonics Technology Letters}, vol.~25, no.~2, pp. 171--174, 2013.

\bibitem{gbps2}
D.~Tsonev, H.~Chun, Rajbhandari \emph{et~al.}, ``A 3-{Gb/s} single-{LED}
  {OFDM}-based wireless {VLC} link using a gallium nitride $\mu$ {LED},''
  \emph{IEEE Photonics Technology Letters}, vol.~26, no.~7, pp. 637--640, 2014.

\bibitem{hospital1}
J.~Miyakoshi, ``Cellular and molecular responses to radio-frequency
  electromagnetic fields,'' \emph{Proceedings of the IEEE}, vol. 101, no.~6,
  pp. 1494--1502, 2013.

\bibitem{hybconf1}
M.~B. Rahaim, A.~M. Vegni, and T.~D. Little, ``A hybrid radio frequency and
  broadcast visible light communication system,'' in \emph{Proc. of IEEE
  GLOBECOM Workshops (GC Wkshps)}, 2011, pp. 792--796.

\bibitem{hybconf2}
D.~A. Basnayaka and H.~Haas, ``Hybrid {RF} and {VLC} systems: Improving user
  data rate performance of {VLC} systems,'' in \emph{IEEE Vehicular Technology
  Conference (VTC Spring)}, 2015, pp. 1--5.

\bibitem{bao2014protocol}
X.~Bao, X.~Zhu, T.~Song, and Y.~Ou, ``Protocol design and capacity analysis in
  hybrid network of visible light communication and {OFDMA} systems,''
  \emph{IEEE Transactions on Vehicular Technology}, vol.~63, no.~4, pp.
  1770--1778, 2014.

\bibitem{kashef2016energy}
M.~Kashef, M.~Ismail, M.~Abdallah, K.~A. Qaraqe, and E.~Serpedin, ``Energy
  efficient resource allocation for mixed {RF/VLC} heterogeneous wireless
  networks,'' \emph{IEEE Journal on Selected Areas in Communications}, vol.~34,
  no.~4, pp. 883--893, 2016.

\bibitem{li2016mobility}
L.~Li, Y.~Zhang, B.~Fan, and H.~Tian, ``Mobility-aware load balancing scheme in
  hybrid {VLC-LTE} networks,'' \emph{IEEE Communications Letters}, vol.~20,
  no.~11, pp. 2276--2279, 2016.

\bibitem{wang2017load}
Y.~Wang, X.~Wu, and H.~Haas, ``Load balancing game with shadowing effect for
  indoor hybrid {LiFi/RF} networks,'' \emph{IEEE Transactions on Wireless
  Communications}, vol.~16, no.~4, 2017.

\bibitem{li2015cooperative}
X.~Li, R.~Zhang, and L.~Hanzo, ``Cooperative load balancing in hybrid visible
  light communications and {WiFi},'' \emph{IEEE Transactions on
  Communications}, vol.~63, no.~4, pp. 1319--1329, 2015.

\bibitem{bao2017visible}
X.~Bao, J.~Dai, and X.~Zhu, ``Visible light communications heterogeneous
  network {(VLC-HetNet)}: new model and protocols for mobile scenario,''
  \emph{Wireless Networks}, vol.~23, no.~1, pp. 299--309, 2017.

\bibitem{rakia2016optimal}
T.~Rakia, H.-C. Yang, F.~Gebali, and M.-S. Alouini, ``Optimal design of
  dual-hop {VLC/RF} communication system with energy harvesting,'' \emph{IEEE
  Communications Letters}, vol.~20, no.~10, pp. 1979--1982, 2016.

\bibitem{chen34downlink}
C.~Chen, D.~A. Basnayaka, and H.~Haas, ``Downlink performance of optical
  attocell networks,'' \emph{Journal of Lightwave Technology}, vol.~34, no.~1,
  pp. 137--156.

\bibitem{shao2015design}
D.~A. Basnayaka and H.~Haas, ``Design and analysis of a hybrid radio frequency
  and visible light communication system,'' \emph{IEEE Transactions on
  Communications}, 2017.

\bibitem{gil}
J.~Gil-Pelaez, ``Note on the inversion theorem,'' \emph{Biometrika}, vol.~38,
  no. 3--4, pp. 481--482, 1951.

\bibitem{hamdi2010useful}
K.~A. Hamdi, ``A useful lemma for capacity analysis of fading interference
  channels,'' \emph{IEEE Transactions on Communications}, vol.~58, no.~2, 2010.

\bibitem{chmodel1}
J.~M. Kahn and J.~R. Barry, ``Wireless infrared communications,''
  \emph{Proceedings of the IEEE}, vol.~85, no.~2, pp. 265--298, 1997.

\bibitem{chmodel2}
T.~Fath and H.~Haas, ``Performance comparison of {MIMO} techniques for optical
  wireless communications in indoor environments,'' \emph{IEEE Transactions on
  Communications}, vol.~61, no.~2, pp. 733--742, 2013.

\bibitem{gammaapprox}
D.~A. Basnayaka and H.~Haas, ``On some inequalities for the incomplete gamma
  function,'' \emph{Mathematics of Computation of the American Mathematical
  Society}, vol.~66, no. 218, pp. 771--778, 1997.

\bibitem{bai2015coverage}
T.~Bai and R.~W. Heath, ``Coverage and rate analysis for millimeter-wave
  cellular networks,'' \emph{IEEE Transactions on Wireless Communications},
  vol.~14, no.~2, pp. 1100--1114, 2015.

\bibitem{andrews2011tractable}
J.~G. Andrews, F.~Baccelli, and R.~K. Ganti, ``A tractable approach to coverage
  and rate in cellular networks,'' \emph{IEEE Transactions on Communications},
  vol.~59, no.~11, pp. 3122--3134, 2011.

\bibitem{noma}
H.~Tabassum, E.~Hossain, and M.~J. Hossain, ``Modeling and analysis of uplink
  non-orthogonal multiple access {(NOMA)} in large-scale cellular networks
  using poisson cluster processes,'' \emph{IEEE Transactions on
  Communications}, 2017.

\bibitem{stoch}
S.~Chiu, D.~Stoyan, W.~Kendall, and J.~Mecke, ``Stochastic geometry and its
  applications,'' \emph{3rd ed. New York: John Wiley and Sons}, 2013.

\bibitem{reflection}
T.~Komine and M.~Nakagawa, ``Fundamental analysis for visiblelight
  communication system using {LED} lights,'' \emph{IEEE Transactions on Cosumer
  Electronics}, vol.~50, no.~1, pp. 100--107, Feb. 2004.

\bibitem{abdi}
A.~Abdi and M.~Kaveh, ``K distribution: An appropriate substitute for
  {R}ayleigh-lognormal distribution in fading-shadowing wireless channels,''
  \emph{IET Electronics Letters}, vol.~34, no.~9, pp. 851--852, 1998.

\end{thebibliography}
\bibliographystyle{IEEEtran}
\end{document}